\documentclass[envcountsame]{llncs}

\usepackage{amsfonts}
\usepackage{amsmath}
\usepackage{amssymb}
\usepackage{color}
\definecolor{orange3}{rgb}{1.0,0.2538,0.1681}
\definecolor{blau}{rgb}{0,.39608,0.74118}
\definecolor{rot}{rgb}{0.79216,.12941,0.24706}

\newcommand\stefan[1]{}
\usepackage{graphics}
\usepackage{psfrag}
\usepackage{tikz}
\usetikzlibrary{arrows,shapes,backgrounds}

\tikzstyle{max}=[thick,draw,minimum size=1.4em,inner sep=0em]
\tikzstyle{min}=[diamond,thick,draw,minimum size=1.4em,%
    inner sep=0em]
\tikzstyle{ran}=[circle,thick,draw,minimum size=1.4em,%
    inner sep=0em]
\tikzstyle{mc}=[rounded corners,thick,draw,minimum size=1.4em,%
    inner sep=.5ex]
\tikzstyle{tran}=[thick,draw,->,>=stealth]
\tikzstyle{loop left}=[tran, to path={.. controls +(150:.5)
    and +(210:.5) .. (\tikztotarget) \tikztonodes}]
\tikzstyle{loop right}=[tran, to path={.. controls +(30:.5)
    and +(330:.5) .. (\tikztotarget) \tikztonodes}]
\tikzstyle{loop above}=[tran, to path={.. controls +(60:.5)
    and +(120:.5) .. (\tikztotarget) \tikztonodes}]
\tikzstyle{loop below}=[tran, to path={.. controls +(240:.5)
    and +(300:.5) .. (\tikztotarget) \tikztonodes}]

\usepackage{microtype}

\newenvironment{qproposition}[2][]{%
{${}$\\[2mm]
\noindent\bf Proposition #2#1.}
\begin{itshape}%
}{%
\end{itshape}%
}
\newcommand{\tran}[1]{{}\mathchoice%
    {\stackrel{#1}{\rightarrow}}
    {\mathop {\smash\rightarrow}\limits^{\vrule width 0pt height 0pt
                                                depth 4pt\smash{#1}}}
    {\stackrel{#1}{\rightarrow}}
    {\stackrel{#1}{\rightarrow}}
{}}
\newcommand{\btran}[1]{{}\mathchoice%
    {\stackrel{#1}{\hookrightarrow}}
    {\mathop {\smash\hookrightarrow}\limits^{\vrule width 0pt height 0pt
                                                depth 4pt\smash{#1}}}
    {\stackrel{#1}{\hookrightarrow}}
    {\stackrel{#1}{\hookrightarrow}}
{}}
\newcommand{\bbtran}[1]{\lhook\joinrel\xrightarrow{#1}}

\newcommand{\bigo}{\mathcal{O}}
\newcommand{\calE}{\mathcal{E}}
\newcommand{\calF}{\mathcal{F}}
\newcommand{\calP}{\mathcal{P}}
\newcommand{\co}[1]{\langle #1 \rangle}

\newcommand{\E}[1]{\mathbb{E}\,[#1]}
\newcommand{\Ex}[1]{\mathbb{E} \left[ #1 \right]}
\newcommand{\Emax}{E_{\it max}}
\newcommand{\fpath}{\mathit{FPath}}

\newcommand{\kw}[1]{{\rmfamily\textbf{{#1}}}}
\newcommand{\massprobcond}[3]{\calP(\mathbf{T}_{#1}=#2 \mid #3)}
\newcommand{\ms}[1]{m^{(#1)}}
\newcommand{\Nset}{\mathbb{N}}
\newcommand{\numbsymb}[1]{\#(#1)}
\newcommand{\Prob}{{\it Prob}}
\newcommand{\pmin}{p_{\it min}}
\newcommand{\Rset}{\mathbb{R}}
\newcommand{\run}{\mathit{Run}}
\newcommand{\tailprob}[2]{\calP(\mathbf{T}_{#1} {\geq} #2)}
\newcommand{\tailprobcond}[3]{\calP(\mathbf{T}_{#1}\geq #2\mid #3)}
\newcommand{\termt}[1]{\mathbf{T}_{#1}}
\newcommand{\thickdot}{\mathop{\raisebox{0.7mm}{\tikz\draw[fill=black] (0,0) circle (1pt);}}} 
\newcommand{\umax}{\vu_{\it max}}
\newcommand{\umin}{\vu_{\it min}}
\newcommand{\vone}{\vec{1}}
\newcommand{\vu}{\vec{u}}
\newcommand{\vv}{\vec{v}}
\newcommand{\Xs}[1]{X^{(#1)}}

\newcommand{\indexbpa}{\bullet}
\newcommand{\bpadelta}{\Delta_\indexbpa}

\newcommand{\theoremlike}[2]{\par\medskip\penalty-250%
{{\bfseries\noindent
#2 \ref{#1}.}}\it}

\newcommand{\thmhelperpre}[2]{\theoremlike{#1}{#2}}
\newcommand{\thmhelperpost}{\par\medskip}

\newenvironment{refproposition}[1]{\thmhelperpre{#1}{Proposition}}{\thmhelperpost }

\begin{document}

\title{Runtime Analysis of Probabilistic Programs with Unbounded Recursion\thanks{%
This work has been published without proofs as a preliminary version
in the \emph{Proceedings of the 38th International Colloquium on Automata, Languages and Programming (ICALP)},
volume 6756 of LNCS, pages 319–-331, 2011 at Springer.
The presentation has been improved since, and the general lower tail bound has been tightened
 from $\Omega(1/n)$ to $\Omega(1/\sqrt{n})$.
}
}

\author{Tom\'{a}\v{s} Br\'{a}zdil\inst{1}$^{\star}$ \and
        Stefan Kiefer\inst{2}$^{\dag}$ \and
        Anton\'{\i}n Ku\v{c}era\inst{1}$^{\star}$ \and
        Ivana~Huta\v{r}ov\'{a}~Va\v{r}ekov\'{a}\inst{1}$^{\ddag}$}

\institute{Faculty of Informatics, Masaryk University, Czech Republic.\\
    \texttt{\{brazdil,kucera\}@fi.muni.cz,ivarekova@centrum.cz} \and
     Department of Computer Science, University of Oxford, United Kingdom.\\
  \texttt{stefan.kiefer@cs.ox.ac.uk}}

\maketitle

\begin{abstract}
\noindent\let\thefootnote\relax\footnote{\makebox[0ex][r]{$^\star~$}%
Tom\'{a}\v{s} Br\'{a}zdil and Anton\'{\i}n Ku\v{c}era are supported
by the Institute
    for Theoretical Computer Science (ITI), project No.~1M0545, and
    by the Czech Science Foundation, grant No.~P202/10/1469.}%
\footnote{\makebox[0ex][r]{$^\dag~$}%
Stefan Kiefer is supported by a postdoctoral fellowship
    of the German Academic Exchange Service (DAAD).}%
\footnote{\makebox[0ex][r]{$^\ddag~$}%
Ivana Huta\v{r}ov\'{a} Va\v{r}ekov\'{a} is supported by by the Czech Science
    Foundation, grant No.~102/09/H042.}%
We study the runtime in
probabilistic programs with unbounded recursion.
As underlying formal model for such programs we use
\emph{probabilistic pushdown automata (pPDA)} which exactly correspond to
recursive Markov chains.
We show that every pPDA can be transformed into a stateless pPDA  (called ``pBPA'')
whose runtime and further properties are closely related to those of the original pPDA.
This result substantially simplifies the analysis of runtime and other pPDA properties.
We prove that
for every pPDA the probability of performing a long run
decreases \emph{exponentially} in the length of the run, if and only
if the expected runtime in the pPDA is \emph{finite}.  If the
expectation is infinite, then the probability decreases
``polynomially''.
We show that these bounds are asymptotically tight.
Our tail bounds on the runtime are \emph{generic}, i.e.,
 applicable to any probabilistic program with unbounded recursion.
An intuitive interpretation is that in pPDA the runtime
is exponentially unlikely to deviate from its expected value.
%
%
%

\end{abstract}

\section{Introduction}
\label{sec-intro}

\noindent
We study the termination time  in programs with
unbounded recursion, which are either randomized or operate on
statistically quantified inputs. As underlying formal
model for such programs we use \emph{probabilistic pushdown
automata (pPDA)} \cite{EKM:prob-PDA-PCTL,EKM:prob-PDA-expectations,%
BEK:prob-PDA-predictions,BBHK:pPDA-discounted}  which are
equivalent to recursive Markov chains
\cite{EY:RMC-SG-equations-JACM,EY:RMC-LTL-complexity,EY:RMC-LTL-QUEST}.
Since pushdown automata are a standard and well-established model for
programs with recursive procedure calls, our abstract results
imply \emph{generic} and \emph{tight} tail bounds for termination time,
the main performance characteristic of probabilistic recursive programs.

A pPDA consists of a finite
set of \emph{control states}, a finite \emph{stack alphabet}, and a finite
set of \emph{rules} of the form $pX \btran{x} q\alpha$, where
$p,q$ are control states, $X$ is a stack symbol, $\alpha$ is a finite
sequence of stack symbols (possibly empty), and $x \in (0,1]$ is
the (rational) probability of the rule. We require that for each $pX$,
the sum of the probabilities of all rules of the form $pX \btran{x} q\alpha$
is equal to~$1$. Each pPDA $\Delta$ induces an infinite-state Markov
chain $M_\Delta$, where the states are configurations of the form
$p\alpha$ ($p$ is the current control state and $\alpha$ is the current
stack content), and $pX\beta \tran{x} q\alpha\beta$ is a transition
of $M_\Delta$ iff $pX \btran{x} q\alpha$ is a rule of~$\Delta$.
We also stipulate that $p\varepsilon \tran{1} p\varepsilon$ for
every control state $p$, where $\varepsilon$ denotes the empty stack.
For example, consider the pPDA $\hat{\Delta}$ with
two control states $p,q$, two stack symbols $X,Y$, and the rules
\[
  pX \bbtran{1/4} p\varepsilon,\
  pX \bbtran{1/4} pXX,\
  pX \bbtran{1/2} qY,\
  pY \bbtran{1}   pY,\
  qY \bbtran{1/2} qX,\
  qY \bbtran{1/2} q\varepsilon,\
  qX \bbtran{1} qY\,.
\]
The structure of Markov chain $M_{\hat{\Delta}}$ is indicated below.

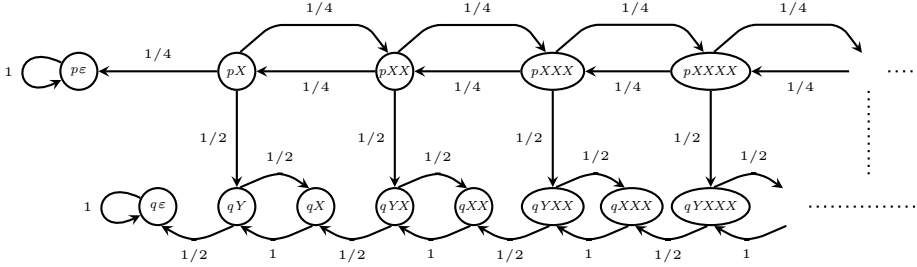
{\centering
\begin{tikzpicture}[x=2.1cm,y=1.8cm,font=\tiny]
  \node (p)        at (0,0)   [ran] {$\mathit{p\varepsilon}$};
  \node (pX)       at (1,0)   [ran] {$\mathit{pX}$};
  \node (pXX)      at (2,0)   [ran] {$\mathit{pXX}$};
  \node (pXXX)     at (3,0)   [ran, shape=ellipse] {$\mathit{pXXX}$};
  \node (pXXXX)    at (4,0)   [ran, shape=ellipse] {$\mathit{pXXXX}$};
  \node (pXXXXX)   at (5,0)   [ran, draw=none] {};
  \node (q)        at (0.5,-1) [ran] {$\mathit{q\varepsilon}$};
  \node (qY)       at (1,-1)   [ran] {$\mathit{qY}$};
  \node (qX)       at (1.5,-1) [ran] {$\mathit{qX}$};
  \node (qYX)      at (2,-1)   [ran] {$\mathit{qYX}$};
  \node (qXX)      at (2.5,-1) [ran] {$\mathit{qXX}$};
  \node (qYXX)     at (3,-1)   [ran, shape=ellipse] {$\mathit{qYXX}$};
  \node (qXXX)     at (3.5,-1) [ran, shape=ellipse] {$\mathit{qXXX}$};
  \node (qYXXX)    at (4,-1)   [ran, shape=ellipse] {$\mathit{qYXXX}$};
  \node (qXXXX)    at (4.5,-1) [ran, shape=ellipse, draw=none] {};
  \draw [tran] (p)     to [loop left]  node[left]  {$1$} (p);
  \draw [tran] (q)     to [loop left]  node[left]  {$1$} (q);
  \draw [tran] (pX)    to node[above] {$1/4$}    (p);
  \draw [tran] (pX)    to node[left]  {$1/2$}    (qY);
  \draw [tran] (pXX)   to node[below] {$1/4$}    (pX);
  \draw [tran] (pXX)   to node[left]  {$1/2$}    (qYX);
  \draw [tran] (pXXX)  to node[below] {$1/4$}    (pXX);
  \draw [tran] (pXXX)  to node[left]  {$1/2$}    (qYXX);
  \draw [tran] (pXXXX) to node[below] {$1/4$}    (pXXX);
  \draw [tran] (pXXXX) to node[left]  {$1/2$}    (qYXXX);
  \draw [tran] (pXXXXX) to node[below] {$1/4$}   (pXXXX);
  \draw [tran, rounded corners] (pX.80) -- +(0.2,0.2) -- node[above] {$1/4$}
          +(0.8,0.2) -- (pXX.110);
  \draw [tran, rounded corners] (pXX.80) -- +(0.2,0.2) -- node[above] {$1/4$}
          +(0.8,0.2) -- (pXXX.110);
  \draw [tran, rounded corners] (pXXX.80) -- +(0.2,0.2) -- node[above] {$1/4$}
          +(0.8,0.2) -- (pXXXX.110);
  \draw [tran, rounded corners] (pXXXX.80) -- +(0.2,0.2) -- node[above] {$1/4$}
          +(0.8,0.2) -- (pXXXXX.110);
  \draw [tran, rounded corners] (qXXXX.260) -- +(-0.2,-0.1) --
          node[below] {$1$} +(-0.3,-0.1) -- (qYXXX.280);
  \draw [tran, rounded corners] (qYXXX.260) -- +(-0.2,-0.1) --
          node[below] {$1/2$} +(-0.3,-0.1) -- (qXXX.280);
  \draw [tran, rounded corners] (qXXX.260) -- +(-0.2,-0.1) --
          node[below] {$1$} +(-0.3,-0.1) -- (qYXX.280);
  \draw [tran, rounded corners] (qYXX.260) -- +(-0.2,-0.1) --
          node[below] {$1/2$} +(-0.3,-0.1) -- (qXX.280);
  \draw [tran, rounded corners] (qXX.260) -- +(-0.2,-0.1) --
          node[below] {$1$} +(-0.3,-0.1) -- (qYX.280);
  \draw [tran, rounded corners] (qYX.260) -- +(-0.2,-0.1) --
          node[below] {$1/2$} +(-0.3,-0.1) -- (qX.280);
  \draw [tran, rounded corners] (qX.260) -- +(-0.2,-0.1) --
          node[below] {$1$} +(-0.3,-0.1) -- (qY.280);
  \draw [tran, rounded corners] (qY.260) -- +(-0.2,-0.1) --
          node[below] {$1/2$} +(-0.3,-0.1) -- (q.280);
  \draw [tran, rounded corners] (qYXXX.80) -- +(0.2,0.1) --
          node[above] {$1/2$} +(0.3,0.1) -- (qXXXX.110);
  \draw [tran, rounded corners] (qYXX.80) -- +(0.2,0.1) --
          node[above] {$1/2$} +(0.3,0.1) -- (qXXX.110);
  \draw [tran, rounded corners] (qYX.80) -- +(0.2,0.1) --
          node[above] {$1/2$} +(0.3,0.1) -- (qXX.110);
  \draw [tran, rounded corners] (qY.80) -- +(0.2,0.1) --
          node[above] {$1/2$} +(0.3,0.1) -- (qX.110);
  \draw [thick, dotted] (pXXXXX) --  +(.3,0);
  \draw [thick, dotted] (qXXXX)  --  +(.8,0);
  \draw [thick, dotted] (pXXXXX) --  +(0,-.8);
\end{tikzpicture}}

pPDA can model programs that use unbounded ``stack-like'' data structures
such as stacks, counters, or even queues (in some cases, the exact
ordering of items stored in a queue is irrelevant and the queue can be
safely replaced with a stack). Transition probabilities may reflect
the random choices
of~the program (such as ``coin flips'' in randomized algorithms) or some
statistical assumptions about the input data.
In particular, pPDA model \emph{recursive} programs.
The global data of such a program are stored
in the finite control, and the individual procedures and functions
together with their local data
correspond to the stack symbols (a function call/return is modeled
by pushing/popping the associated stack symbol onto/from the stack).
%
%
As a simple example, consider the
recursive program \emph{Tree} of Figure~\ref{fig:tree}, which
computes the value of an \emph{And/Or-tree}, i.e., a tree such
that (i) every node has either zero or two children,
(ii) every inner node is either an And-node or an Or-node, and (iii)
on any path from the root to a leaf And- and Or-nodes alternate. We
further assume that the root is either a leaf or an And-node. \emph{Tree}
starts by invoking the function \texttt{And} on the root of a given
And/Or-tree.
Observe that the program evaluates subtrees only if necessary.
Now assume that the input are random And/Or trees following
the Galton-Watson distribution: a node of the tree has two children
with probability~$1/2$, and no children with probability~$1/2$.
Furthermore, the conditional probabilities that a childless node
evaluates to $0$ and $1$ are also both equal to $1/2$. On inputs with this
distribution, the algorithm corresponds to a pPDA $\Delta_{\mathit{Tree}}$
of Figure~\ref{fig:tree} (the control states $r_0$ and $r_1$ model
the return values $0$ and $1$).

\newcommand{\AndOrProg}[3]{%
\parbox[t]{.1\textwidth}{\flushleft\ttfamily%
\begin{tabbing}
\hspace*{.7em}\=\hspace*{.7em}\=\hspace*{.7em}\=\kill
    \kw{function} #1(node) \\
    \> \kw{if} node.leaf \kw{then} \\
    \>\> \kw{return} node.value \\
    \> \kw{else} \\
    \>\> $v$ := #2(node.left) \\
    \>\> \kw{if} $v = #3$ \kw{then} \\
    \>\>\> \kw{return} $#3$ \\
    \>\> \kw{else} \\
    \>\>\> \kw{return} #2(node.right)
\end{tabbing}}}
\begin{figure}[t]
\begin{tabular}{c@{\hspace{20mm}}c}
\AndOrProg{And}{Or}{0} & \AndOrProg{Or}{And}{1} \\
\parbox[t]{.3\textwidth}{\flushleft
\begin{align*}
qA &\btran{1/4} r_1\varepsilon \\
qA &\btran{1/4} r_0\varepsilon \\
qA &\btran{1/2} qOA \\
r_0A &\btran{1}  r_0\varepsilon \\
r_1A &\btran{1}  qO
\end{align*}
}
&
\parbox[t]{.3\textwidth}{\flushleft
\begin{align*}
qO &\btran{1/4} r_1\varepsilon \\
qO &\btran{1/4} r_0\varepsilon \\
qO &\btran{1/2} qAO \\
r_1O &\btran{1} r_1\varepsilon \\
r_0O &\btran{1} qA
\end{align*}
}
\end{tabular}
\caption{The program \emph{Tree} and its pPDA model $\Delta_{\mathit{Tree}}$.}
\label{fig:tree}
\end{figure}

We study the
\emph{termination time} of runs in
a given pPDA $\Delta$.
  For every pair of control states $p,q$
  and every stack symbol $X$ of $\Delta$, let $\run(pXq)$ be
  the set of all runs
  (infinite paths) in $M_\Delta$ initiated in $pX$ which visit
  $q\varepsilon$. The termination time is modeled by the
  random variable $\termt{pX}$, which to every run~$w$
  assigns either the number of steps needed to reach
  a configuration with empty stack, or $\infty$ if there is no such
  configuration.
  The conditional expected value \mbox{$\E{\termt{pX} \mid \run(pXq)}$},
  denoted just by $E[pXq]$ for short,
  then corresponds to the average number of steps needed to reach
  $q\varepsilon$ from $pX$, computed only for those runs initiated
  in $pX$ which terminate in $q\varepsilon$.
  For example, using the results of
  \cite{EKM:prob-PDA-PCTL,EKM:prob-PDA-expectations,EY:RMC-SG-equations-JACM}, one can show that the functions
  \texttt{And} and \texttt{Or}
  of the program \emph{Tree} terminate with probability one, and
  the expected termination times can be computed by solving
  a system of linear equations.
  Thus, we obtain the following:
  \begin{align*}
     E[q A r_0] &= 7.155113 & E[q A r_1] &= 7.172218 \\
     E[q O r_0] &= 7.172218 & E[q O r_1] &= 7.155113 \\
     E[r_0 A r_0] &= 1.000000 & E[r_1 A r_0] &= 8.172218 & E[r_1 A r_1] &= 8.155113 \\
     E[r_1 O r_1] &= 1.000000 & E[r_0 O r_1] &= 8.172218 & E[r_0 O r_0] &=  8.155113
  \end{align*}
  However, the mere expectation of the termination time does not provide much
  information about its distribution 
  until we
  analyze the associated \emph{tail bound}, i.e., the probability that
  the termination time deviates from its expected value by a given
  amount. That is, we are interested in bounds for the conditional
  probability \mbox{$\tailprobcond{pX}{n}{\run(pXq)}$}.
  (Note this probability makes sense regardless of whether $E[pXq]$ is finite or infinite.)
  Assuming that the (conditional) expectation and variance of $\termt{pX}$
  are finite, one can apply Markov's and Chebyshev's inequalities and
  thus yield bounds of the form
  \mbox{$\tailprobcond{pX}{n}{\run(pXq)} \le c/n$} and
  \mbox{$\tailprobcond{pX}{n}{\run(pXq)} \le c/{n^2}$}, respectively,
  where $c$ is a constant depending only on the underlying pPDA.
  However, these bounds are asymptotically always worse than
  our exponential bound (see below). If $E[pXq]$ is infinite, these
  inequalities cannot be used at all.

\smallskip

\noindent
\textbf{Our contribution.}\quad
The main contributions of this paper are the following:
\begin{itemize}
\item We show that every pPDA can be effectively transformed into
  a \emph{stateless}
  pPDA (called ``pBPA'') so that all important quantitative characteristics
  of runs are preserved. This simple (but fundamental) observation was
  overlooked in previous works
  on pPDA and related models \cite{EKM:prob-PDA-PCTL,EKM:prob-PDA-expectations,%
  BEK:prob-PDA-predictions,BBHK:pPDA-discounted,EY:RMC-SG-equations-JACM,%
  EY:RMC-LTL-complexity,EY:RMC-LTL-QUEST}, although it simplifies virtually
  all of these results.
  Hence, we can w.l.o.g.\ concentrate just on the study of pBPA.
  Moreover, for the runtime analysis, the transformation yields a pBPA all of whose symbols
   terminate with probability~one, which further simplifies the analysis.
\item 
  We provide tail bounds for $\termt{pX}$ which are
  \emph{asymptotically optimal for every pPDA} and are applicable
  also in the case when $E[pXq]$
  is infinite. More precisely, we show that for every pair
  of control states $p,q$ and every stack symbol $X$, there
  are essentially three possibilities:
  \begin{itemize}
  \item There is a ``small'' $k$ such that
    $\tailprobcond{pX}{n}{\run(pXq)} = 0$ for all $n \geq k$.
  \item $E[pXq]$ is finite and
    $\tailprobcond{pX}{n}{\run(pXq)}$ decreases exponentially in~$n$.
  \item $E[pXq]$ is infinite and
    $\tailprobcond{pX}{n}{\run(pXq)}$ decreases ``polynomially'' in~$n$.
  \end{itemize}
  The exact formulation of this result, including the explanation of
  what is meant by a ``polynomial'' decrease, is given in
  Theorem~\ref{thm:termination} (technically, Theorem~\ref{thm:termination}
  is formulated for pBPA which terminate with probability one, which
  is no restriction as explained above). Observe that a direct consequence
  of the above theorem is that \emph{all} conditional moments
  \mbox{$\E{\termt{pX}^k \mid \run(pXq)}$} are simultaneously either
  finite or infinite (in particular, if $E[pXq]$ is finite, then
  so is the conditional variance of~$\termt{pX}$).
\end{itemize}
The characterization given in Theorem~\ref{thm:termination} is
effective. In particular, it is decidable in polynomial space whether
$E[pXq]$ is finite or infinite by using the results of
\cite{EKM:prob-PDA-PCTL,EKM:prob-PDA-expectations,%
  EY:RMC-SG-equations-JACM}, and if $E[pXq]$ is finite, we can compute
concrete bounds on the probabilities.  Our results vastly improve on
what was previously known on the termination time~$\termt{pX}$.
Previous work, in
particular~\cite{EKM:prob-PDA-expectations,Brazdil:PhD}, has focused
on computing expectations and variances for a class of random
variables on pPDA runs, a class that includes $\termt{pX}$ as prime
example.  Note that our exponential bound given in
Theorem~\ref{thm:termination} depends, like Markov's
inequality, only on expectations, which can be efficiently
approximated by the methods
of~\cite{EKM:prob-PDA-expectations,EKL10:SICOMP}.

An intuitive interpretation of our results is that pPDA with finite
(conditional) expected termination time are well-behaved in the sense
that the termination time is exponentially unlikely to deviate from
its expectation. Of course, a detailed analysis of a concrete pPDA
may lead to better bounds, but these bounds
will be \emph{asymptotically equivalent} to our generic bounds.
Also note that the conditional expected termination time can be finite
even for pPDA that do not terminate with probability one. Hence, for
every $\varepsilon > 0$ we can compute a tight threshold $k$ such that
if a given pPDA terminates at all, it terminates after at most
$k$~steps with probability $1-\varepsilon$ (this is useful for
interrupting programs that are supposed but not guaranteed to
terminate).

\smallskip

\noindent
\textbf{Proof techniques.}
The main mathematical tool for establishing our results on runtime is
(basic) martingale theory and its tools such as the
optional stopping theorem and Azuma's inequality (see Section~\ref{sec:BPA-analysis}).
More precisely, we construct two different martingales corresponding to the cases
 when the expected termination time is finite resp.\ infinite.
In combination with our reduction to pBPA
 this establishes a powerful link between pBPA, pPDA, and martingale theory.

Our analysis of termination time in the case when the
expected termination time is infinite builds on Perron-Frobenius
theory for nonnegative matrices as well as on recent results
from~\cite{EY:RMC-SG-equations-JACM,EKL10:SICOMP}.  We also use some
of the observations presented in
\cite{EKM:prob-PDA-PCTL,EKM:prob-PDA-expectations,%
BEK:prob-PDA-predictions}.
\smallskip

\noindent
\textbf{Related work.}
The application of Azuma's inequality in the analysis of particular randomized algorithms
is also known as the \emph{method of bounded differences}; see, e.g., \cite{MR:book,DP:book} and the references therein.
In contrast, we apply martingale methods not to particular algorithms, but to the pPDA model as a whole.

Analyzing the distribution of termination time is closely related to
 the analysis of multitype branching processes (MT-BPs)~\cite{Harris:book}.
A MT-BP is very much like a pBPA (see above).
The stack symbols in pBPA correspond to species in MT-BPs.
An $\varepsilon$-rule corresponds to the death of an individual,
 whereas a rule with two or more symbols on the right hand side
 corresponds to reproduction.
Since in MT-BPs the symbols on the right hand side of rules
 evolve concurrently, termination time in pBPA does {\em not} correspond to
 extinction time in MT-BPs, but to the size of the {\em total
 progeny} of an individual, i.e., the number of direct or
 indirect descendants of an individual.
 The distribution of the total progeny of a MT-BP has been studied
 mainly for the case of a single species,
 see, e.g.,~\cite{Harris:book,Pakes71,QuineS94} and the references therein,
 but to the best of our knowledge, no tail bounds for MT-BPs
 have been given. Hence, Theorem~\ref{thm:termination} can also be
 seen as a contribution to MT-BP theory.

Stochastic context-free grammars (SCFGs)
\cite{MS:book} are also closely related to pBPA.
The termination time in pBPA corresponds to the number of nodes in a derivation tree
 of a SCFG, so our analysis of pBPA immediately applies to SCFGs.
Quasi-Birth-Death processes (QBDs) can also be seen as a special case
of pPDA. A QBD is a generalization of a birth-death
process 
studied in queueing theory and applied
probability (see, e.g., \cite{LR:book,BLM:book,EWY:one-counter}).
Intuitively, a QBD describes an unbounded queue, using a counter to count
the number of jobs in the queue, where the queue can be in one of
finitely many distinct ``modes''. Hence, a (discrete-time) QBD can be
equivalently defined by a pPDA with one stack symbol used to emulate
the counter. These special pPDA are also known as \emph{probabilistic
one-counter automata (pOC)} \cite{EWY:one-counter,BBEKW:OC-MDP,BBE:OC-games}.
Recently, it has been shown in \cite{BKK:pOC-time-LTL-martingale-arxiv}
that every pOC induces a martingale apt for studying the properties
of both terminating and nonterminating runs in pOC. The construction
is based on ideas specific to pOC that are completely unrelated to
the ones presented in this paper.

Previous work on pPDA and the equivalent model of recursive Markov chains
includes \cite{EKM:prob-PDA-PCTL,EKM:prob-PDA-expectations,%
BEK:prob-PDA-predictions,BBHK:pPDA-discounted,EY:RMC-SG-equations-JACM,%
EY:RMC-LTL-complexity,EY:RMC-LTL-QUEST}. In this paper we use many
of the results presented in these papers, which is explicitly acknowledged
at appropriate places.
\smallskip

\noindent\textbf{Organization of the paper.}
We present our results after some preliminaries in Section~\ref{sec:prelim}.
In Section~\ref{sec:transformation} we show how to transform a given pPDA into an equivalent pBPA,
and in Section~\ref{sec:BPA-analysis} we design the promised martingales and derive tight tail bounds for the termination time.
We conclude in Section~\ref{sec:concl}.
Some proofs have been moved to Section~\ref{sec:proofs}.


\section{Preliminaries} \label{sec:prelim}
\noindent
In the rest of this paper, $\Nset$, $\Nset_0$,
and $\Rset$ denote the set of positive integers, non-negative integers,
and real numbers, respectively. The tuples of
$A_1 \times A_2 \cdots \times A_n$ are often written simply as
$a_1a_2\dots a_n$.
%
The set of all finite words over a given alphabet $\Sigma$ is denoted by
$\Sigma^*$, and the set of all infinite words over $\Sigma$ is denoted by
$\Sigma^\omega$.
We 
write $\varepsilon$ for the empty word.
The length of a given $w \in \Sigma^* \cup
\Sigma^{\omega}$ is denoted by $|w|$, where the length of an infinite
word is $\infty$.
Given a word (finite or infinite) over $\Sigma$, the individual
letters of $w$ are denoted by $w(0),w(1),\dots$
For $X\in \Sigma$ and $w \in \Sigma^*$, we denote by
  $\numbsymb{X}(w)$ the number of occurrences of $X$ in~$w$.

%
%
%
%
\begin{definition}[\textbf{Markov Chains}]
\label{def-MC}
  A \emph{Markov chain} is a triple \mbox{$M = (S,\tran{},\Prob)$}
  where $S$ is a finite or countably infinite set of \emph{states},
  ${\tran{}} \subseteq S \times S$ is a \emph{transition relation},
  and $\Prob$ is a function which to each transition $s \tran{} t$ of
  $M$ assigns its probability $\Prob(s \tran{} t) > 0$ so that
  for every $s \in S$ we have $\sum_{s \rightarrow t} \Prob(s \tran{} t) =
  1$ (as usual, we write $s \tran{x} t$ instead of $\Prob(s \tran{} t) = x$).
\end{definition}

\noindent
A \emph{path} in $M$ is a finite or infinite word $w \in S^+ \cup S^\omega$
such that $w(i{-}1) \tran{} w(i)$ for every \mbox{$1 \leq i < |w|$}.
For a state~$s$, we use $\fpath(s)$ to denote the set of all finite paths initiated in~$s$.
A \emph{run} in $M$ is an infinite path in~$M$.
We denote by $\run[M]$ the set of all runs in~$M$. The set of all runs
that start with a given finite path $w$ is denoted by $\run[M](w)$.
When $M$ is understood, we write just $\run$ and $\run(w)$ instead of
$\run[M]$ and $\run[M](w)$, respectively. Given $s\in S$ and $A\subseteq S$,
we say \emph{$A$ is reachable from~$s$} if there is a run $w$ such that
$w(0)=s$ and $w(i) \in A$ for some $i\geq 0$.

To every $s \in S$ we associate the probability
space $(\run(s),\calF,\calP)$ where
$\calF$ is the \mbox{$\sigma$-field} generated by all \emph{basic cylinders}
$\run(w)$ where $w$ is a finite path starting with~$s$, and
$\calP: \calF \rightarrow [0,1]$ is the unique probability measure such that
$\calP(\run(w)) = \Pi_{i{=}1}^{|w|-1} x_i$ where
$w(i{-}1) \tran{x_i} w(i)$ for every $1 \leq i < |w|$.
If $|w| = 1$, we put $\calP(\run(w)) = 1$. Note that only certain
subsets of $\run(s)$ are $\calP$-measurable, but in this paper we only
deal with ``safe'' subsets that are guaranteed to be in $\calF$.
%

%
%
\begin{definition}[\textbf{probabilistic PDA}]
\label{def-pPDA}
  A \emph{probabilistic pushdown automaton (pPDA)} is a tuple $\Delta =
  (Q,\Gamma,{\btran{}},\Prob)$ where $Q$ is a finite set of \emph{control
    states}, $\Gamma$ is a finite \emph{stack alphabet}, ${\btran{}}
  \subseteq (Q \times \Gamma) \times (Q \times \Gamma^{\leq 2})$
  is a \emph{transition relation}
  (where $\Gamma^{\leq 2}=\{\alpha\in
  \Gamma^*, |\alpha|\leq 2\}$), and $\Prob$ is a function which to
    each transition \mbox{$pX \btran{} q\alpha$} assigns its
  probability \mbox{$\Prob(pX \btran{} q\alpha) > 0$} so that for
  all $p \in Q$ and $X \in \Gamma$ we have that $\sum_{pX \hookrightarrow
    q\alpha} \Prob(pX \btran{} q\alpha) = 1$.
    As usual, we write $pX \btran{x} q\alpha$
    instead of $\Prob(pX \btran{} q\alpha) = x$.
\end{definition}

\noindent
Elements of $Q\times \Gamma^*$ are called~\emph{configurations}
of $\Delta$. A pPDA with just one control state is
called pBPA.\footnote{The ``BPA'' acronym stands for
``Basic Process Algebra'' and it is used mainly for historical
reasons. pBPA are closely related to stochastic context-free grammars
and are also called \emph{\mbox{1-exit} recursive Markov chains}
(see, e.g., \cite{EY:RMC-SG-equations-JACM}).}
In what follows, configurations of pBPA are usually written without
the (only) control state~$p$ (i.e., we write just $\alpha$
instead~of~$p\alpha$).
%
%
We define the \emph{size} of a pPDA $\Delta$ as
$|\Delta|=|Q|+|\Gamma|+|{\btran{}}|+ |{\Prob}|$, where $|{\Prob}|$ is
the sum of sizes of binary representations of values taken by ${\Prob}$.
To $\Delta$ we associate the Markov chain $M_\Delta$ with 
$Q \times \Gamma^*$ as the set of states and transitions defined as follows:
\begin{itemize}
\item $p \varepsilon \tran{1} p \varepsilon$ for each $p \in Q$;
\item $pX \beta \tran{x} q \alpha \beta$ is a transition of
  $M_\Delta$ if{}f $pX \btran{x} q\alpha$ is a transition of $\Delta$.
\end{itemize}
For all $pXq \in Q \times \Gamma \times Q$ and $rY \in Q \times \Gamma$,
we define
\begin{itemize}
\item $\run(pXq)=\{w\in \run(pX)\mid  w(i)=q\varepsilon \mbox{ for some }
  i \in \Nset\}$
\item $\run(rY{\uparrow})=\run(rY)\setminus \bigcup_{s\in Q} \run(rYs)$.
\end{itemize}
Further, we put $[pXq]=\calP(\run(pXq))$ and
$[pX{\uparrow}]=\calP(\run(pX{\uparrow}))$. If $\Delta$ is a pBPA, we write
$[X]$ and $[X{\uparrow}]$ instead of $[pXp]$ and $[pX{\uparrow}]$,
where $p$ is the only control state of~$\Delta$.

Let $p\alpha \in Q \times \Gamma^*$.
We denote by $\termt{p\alpha}$ a random variable over $\run(p\alpha)$ where
$\termt{p\alpha}(w)$ is either the least $n \in \Nset_0$ such that
$w(n) = q\varepsilon$ for some $q \in Q$, or $\infty$ if there is no
such~$n$. Intuitively, $\termt{p\alpha}(w)$ is the number of steps
(``the time'') in which the run $w$ initiated in~$p\alpha$ terminates.
We write $E[p \alpha] := \Ex{\termt{p \alpha}}$ for the expected termination time
 (usually omitting the control state~$p$ for pBPA).


\section{Transforming pPDA into pBPA} \label{sec:transformation}

Let $\Delta = (Q,\Gamma,{\btran{}},\Prob)$ be a pPDA. We show how to
construct a pBPA $\bpadelta$ which is ``equivalent'' to $\Delta$ in a
well-defined sense. This construction is a relatively straightforward
modification of the standard method for transforming a PDA into an
equivalent context-free grammar (see, e.g., \cite{HU:book}), but
has so far been overlooked in the existing literature on
probabilistic PDA.  The idea behind this method is to construct
a BPA with stack symbols of the form $\co{pXq}$ for all $p,q\in Q$ and
$X\in\Gamma$.  Roughly speaking, such a triple corresponds to
terminating paths from $pX$ to $q\varepsilon$.  Subsequently,
transitions of the BPA are induced by transitions of the PDA in a way
corresponding to this intuition. For example, a transition of the form
$pX\btran{} rYZ$ induces transitions of the form $\co{pXq}\btran{}
\co{rYs}\co{sZq}$ for all $s\in Q$. Then each path from $pX$ to
$q\varepsilon$ maps naturally to a path from $\co{pXq}$ to~$\varepsilon$.  
This construction can also be applied in the
probabilistic setting by assigning probabilities to transitions so that
the probability of the corresponding paths is preserved.
We also deal with nonterminating runs by introducing new stack symbols
of the form $\co{pX{\uparrow}}$.

Formally,
the stack alphabet
of $\bpadelta$  is defined as follows: For every $pX \in Q \times \Gamma$
such that $[pX{\uparrow}] >0$ we add a stack symbol $\co{pX{\uparrow}}$,
and for every $pXq \in Q\times\Gamma\times Q$ such that $[pXq]>0$
we add a stack symbol $\co{pXq}$. Note that the stack alphabet
of $\bpadelta$ is effectively
constructible in polynomial space by applying the results of
\cite{EKM:prob-PDA-PCTL,EY:RMC-SG-equations-JACM}.

Now we construct the rules $\bbtran{}_{\indexbpa}$ of~$\bpadelta$.
For all $\co{pXq}$ we have the following rules:
\begin{itemize}
\item if $pX \btran{x} rYZ$ in $\Delta$, then for all $s\in Q$ such that
 $y = x \cdot [rYs] \cdot [sZq] > 0$
  we put
  $\langle pXq\rangle \bbtran{y/[pXq]}_\indexbpa \langle rYs\rangle\langle sZq\rangle$;
\item if $pX \btran{x} rY$ in $\Delta$, where $y = x \cdot[rYq]>0$, we put
  $\langle pXq\rangle \bbtran{y/[pXq]}_\indexbpa \langle rYq\rangle$;
\item if $pX \btran{x} q\varepsilon$ in $\Delta$, we put
  $\langle pXq\rangle\bbtran{x/[pXq]}_\indexbpa \varepsilon$.
\end{itemize}

For all $\co{pX{\uparrow}}$ we have the following rules:

\begin{itemize}
\item if $pX \btran{x} rYZ$ in $\Delta$, then for every $s \in Q$ where $y = x\cdot [rYs]\cdot [sZ{\uparrow}] > 0$
    we add $\co{pX{\uparrow}} \bbtran{y/[pX{\uparrow}]}_\indexbpa \co{rYs}\co{sZ{\uparrow}}$;
\item for all $qY \in Q \times \Gamma$ where
   $x = [qY{\uparrow}] \cdot \sum_{pX \hookrightarrow qY\beta} \Prob(pX \btran{} qY\beta) >0$, we add
   $\co{pX{\uparrow}} \bbtran{x/[pX{\uparrow}]}_\indexbpa \co{qY{\uparrow}}$.
\end{itemize}
Note that the transition probabilities of~$\bpadelta$ may take irrational
values. Still, the construction of $\bpadelta$ is to some extent
``effective'' due to the following proposition:
\begin{proposition}[\cite{EKM:prob-PDA-PCTL,EY:RMC-SG-equations-JACM}]
\label{prop:pdatobpa-effective}
 Let $\Delta = (Q,\Gamma,{\btran{}},\Prob)$ be a pPDA.
 Let $pXq\in Q\times \Gamma \times Q$.
 There is a formula $\Phi(x)$ of
 $\mathit{ExTh(\mathbb{R})}$ (the existential theory of the reals)
 with one free variable $x$ such that the length of $\Phi(x)$ is polynomial
 in $|\Delta|$ and $\Phi(x/r)$ is valid iff $r = [pXq]$.
\end{proposition}
\noindent
Using Proposition~\ref{prop:pdatobpa-effective}, one can compute
formulae of $\mathit{ExTh(\mathbb{R})}$ that ``encode'' transition
probabilities of $\bpadelta$. Moreover, these probabilities can be
effectively approximated up to an arbitrarily small error by
employing either the decision procedure for $\mathit{ExTh(\mathbb{R})}$
\cite{Canny:Tarski-exist-PSPACE} or by
using Newton's method \cite{EKL08:stacs,KLE07:stoc,EKL10:SICOMP}.

\begin{example}
Consider a pPDA $\Delta$ with two control states, $p,q$, one stack symbol, $X$, and the following transition rules:
{\small
\[
  pX \bbtran{a} qXX,\
  pX \bbtran{1-a} q\varepsilon,\
  qX \bbtran{b} pXX,\
  qX \bbtran{1-b}   p\varepsilon,\
\]}%
where both $a,b$ are greater than $1/2$.
Apparently, $[pXp]=[qXq]=0$. Using results of~\cite{EKM:prob-PDA-PCTL} one can easily verify that $[pXq]=(1-a)/b$ and $[qXp]=(1-b)/a$.
Thus $[pX{\uparrow}]=(a+b-1)/b$ and $[qX{\uparrow}]=(a+b-1)/a$. Thus the stack symbols of $\bpadelta$
are $\co{pXq},\co{qXp},\co{pX{\uparrow}},\co{qX{\uparrow}}$. The transition rules of $\bpadelta$ are:
{\small
\[
\begin{array}{llll}
  \co{pXq} \bbtran{1-b}_\indexbpa \co{qXp}\co{pXq} &
  \quad \co{pXq} \bbtran{b}_\indexbpa \varepsilon &
  \quad \co{qXp} \bbtran{1-a}_\indexbpa \co{pXq}\co{qXp} &
  \quad \co{qXp} \bbtran{a}_\indexbpa \varepsilon \\
  \co{pX{\uparrow}} \bbtran{1-b}_\indexbpa \co{qXp}\co{pX{\uparrow}} &
  \quad \co{pX{\uparrow}} \bbtran{b}_\indexbpa \co{qX{\uparrow}} &
  \quad \co{qX{\uparrow}} \bbtran{1-a}_\indexbpa \co{pXq}\co{qX{\uparrow}} &
  \quad \co{qX{\uparrow}} \bbtran{a}_\indexbpa \co{pX{\uparrow}}
\end{array}
\]}%
As both $a,b$ are greater than $1/2$, the resulting pBPA has a tendency to remove symbols rather than add symbols. Thus both $\co{pXq}$ and $\co{qXp}$ terminate with probability $1$.
\end{example}
\noindent
When studying long-run properties of pPDA (such as $\omega$-regular
properties or limit-average properties), one usually assumes that
the runs are initiated in a configuration $p_0X_0$ which cannot
terminate, i.e., $[p_0X_0{\uparrow}] = 1$. Under this assumption,
the probability spaces over $\run[M_{\Delta}](p_0 X_0)$ and
$\run[M_{\bpadelta}](\co{p_0 X_0 {\uparrow}})$ are ``isomorphic''
w.r.t.{} all properties that depend only on the control states
and the top-of-the-stack symbols of the configurations visited
along a run. This is formalized in our next proposition.

\begin{proposition}\label{lem:longrun_PDA_to_BPA}
Let $p_0 X_0 \in Q \times \Gamma$ such that $[p_0 X_0 {\uparrow}] = 1$. Then
there is a partial function
$\Upsilon:\run[M_{\Delta}](p_0 X_0)\rightarrow
\run[M_{\bpadelta}](\co{p_0 X_0 {\uparrow}})$ such that for every
$w \in \run[M_{\Delta}](p_0 X_0)$, where $\Upsilon(w)$ is defined, and every
$n \in \Nset$ we have the following: if $w(n) = qY\beta$, then
$\Upsilon(w)(n) = \co{qY{\dag}}\gamma$, where $\dag$ is either an element
of $Q$ or ${\uparrow}$. Further, for every measurable set of runs
$R \subseteq \run[M_{\bpadelta}](\co{p_0 X_0 {\uparrow}})$ we have that
$\Upsilon^{-1}(R)$ is measurable and $\calP(R)=\calP(\Upsilon^{-1}(R))$.
\end{proposition}

\noindent
As for terminating runs, observe that the ``terminating''
symbols of the form $\langle pXq\rangle$ do not depend on the
``nonterminating'' symbols of the form $\co{pX{\uparrow}}$, i.e.,
if we restrict $\bpadelta$ just to terminating symbols, we again
obtain a pBPA.
%
A straightforward computation reveals the following
proposition about terminating runs that is crucial for our results
presented in the next section.
\begin{proposition}\label{prop:pdatobpa}
  Let $pXq\in Q\times \Gamma\times Q$ and $[pXq]>0$.  Then almost all
  runs of $M_{\bpadelta}$ initiated in $\langle pXq\rangle$ terminate, i.e., reach
  $\varepsilon$.  Further, for all $n\in \Nset$ we have that
  \[
    \massprobcond{pX}{n}{\run(pXq)} \quad = \quad
    \massprobcond{\langle pXq\rangle}{n}{\run(\langle pXq \rangle)}
  \]
\end{proposition}
Observe that this proposition, together with a very special form of
rules in $\bpadelta$, implies that all configurations reachable from a
nonterminating configuration $p_0 X_0$ have the form $\alpha
\co{qY{\uparrow}}$, where $\alpha$ terminates almost surely and
$\co{qY{\uparrow}}$ never terminates. It follows that such a pBPA can
be transformed into a finite-state Markov chain (whose states are
the nonterminating symbols) which is allowed to make recursive calls
that almost surely
terminate (using rules of the form $\co{pX{\uparrow}}\bbtran{}
\co{rZq}\co{qY{\uparrow}}$). This observation is very useful
when investigating the properties of nonterminating runs,
and many of the existing results about pPDA can be substantially
simplified using this result.

\section{Analysis of pBPA} \label{sec:BPA-analysis}

In this section we establish the promised tight tail bounds for
the termination time. By virtue of Proposition~\ref{prop:pdatobpa},
it suffices to analyze \emph{almost surely terminating} pBPA,
 i.e., pBPA all whose stack symbols terminate with probability~$1$.  In
what follows we assume that $\Delta$ is such a pBPA, and we also
fix an initial stack symbol~$X_0$.
For $X,Y \in \Gamma$, we say that \emph{$X$ depends directly on~$Y$},
if there is a
rule $X \btran{} \alpha$ such that $Y$ occurs in~$\alpha$.  Further,
we say that \emph{$X$ depends on~$Y$}, if either $X$ depends directly
on~$Y$, or $X$ depends directly on a symbol~$Z \in \Gamma$ which
depends on~$Y$.  One can compute, in linear time, the directed acyclic
graph (DAG) of strongly connected components (SCCs) of the dependence
relation.  The \emph{height} of this DAG, denoted by $h$, is
defined as the longest
distance between a top SCC and a bottom SCC plus~$1$ (i.e., $h=1$ if
there is only one SCC). We can safely assume that all symbols on
which $X_0$ does not depend were removed from~$\Delta$.
We abbreviate $\tailprobcond{X_0}{n}{\run(X_0)}$ to $\tailprob{X_0}{n}$,
and we use $\pmin$ to denote $\min\{ p \mid X\btran{p}\alpha\text{ in }\Delta\}$.
%
%
Here is our main result:
\begin{theorem}\label{thm:termination}
Let $\Delta$ be an almost surely terminating pBPA with stack alphabet~$\Gamma$.
Assume that $X_0 \in \Gamma$ depends on
all $X \in \Gamma \setminus \{X_0\}$, and let
$\pmin =\min\{ p \mid X\btran{p}\alpha\text{ in }\Delta\}$.
Then one of the following is true:
\begin{enumerate}
\item[(1)] $\tailprob{X_0}{2^{|\Gamma|}}=0$.
\item[(2)] $E[X_0]$ is {\bf finite} and for all $n\in \Nset$
with $n \ge 2 E[X_0]$ we have that
\begin{equation*}\textstyle
 \pmin^{n}\quad \leq\quad \tailprob{X_0}{n} \quad \leq\quad
 \exp\left( 1 - \frac{n}{8\Emax^2} \right)
\end{equation*}
where $\Emax = \max_{X \in \Gamma} E[X]$.
\item[(3)] $E[X_0]$ is {\bf infinite} and there is
$n_0 \in \Nset$ such that for all $n \ge n_0$ we have that
\begin{equation*}\textstyle
c/n^{1/2} \quad\le\quad \tailprob{X_0}{n} \quad\le\quad d_1 / n^{d_2}
\end{equation*}
where $d_1 = 18 h |\Gamma| / \pmin^{3|\Gamma|}$, and $d_2 = {1/(2^{h+1}-2)}$.
Here, $h$ is the height of the DAG of SCCs of the dependence relation,
 and $c$ is a suitable positive constant depending on~$\Delta$.
\end{enumerate}
\end{theorem}

\noindent
More colloquially, Theorem~\ref{thm:termination} states that
 $\Delta$ satisfies either (1) or (2) or (3), where (1) is when $\Delta$ does not have any long terminating runs;
 and (2) resp.~(3) is when the expected termination time is finite (resp.~infinite) and the probability of
   performing a terminating run of length~$n$ decreases exponentially (resp.~polynomially) in~$n$.

One can effectively distinguish between the three cases
set out in Theorem~\ref{thm:termination}. More precisely, case (1)
can be recognized in polynomial time by looking only at the structure
of the pBPA, i.e., disregarding the probabilities.  Determining whether
$E[X_0]$ is finite or infinite can be done in polynomial space
by employing the results of \cite{EKM:prob-PDA-expectations,Brazdil:PhD}.
This holds even if the transition probabilities of~$\Delta$ are represented just
symbolically by formulae of $\mathit{ExTh(\mathbb{R})}$
(see Proposition~\ref{prop:pdatobpa-effective}).

The proof of Theorem~\ref{thm:termination} is based on designing suitable martingales
 that are used to analyze the concentration of the termination time. Recall
that  a \emph{martingale} is an infinite sequence of random variables
$\ms{0},\ms{1},\dots$ such that, for all $i \in \Nset$,
$\E{|\ms{i}|} < \infty$, and
\mbox{$\E{\ms{i+1} \mid \ms{1},\dots,\ms{i}} = \ms{i}$} almost surely.
If $|\ms{i}-\ms{i-1}| < c_i$ for all $i \in \Nset$, then we have the
following \emph{Azuma's inequality} (see, e.g.,~\cite{book:Williams}):
\[
  \calP(\ms{n} - \ms{0} \geq t)  \quad \leq \quad
  \exp\left( \frac{-t^2}{2\sum_{k=1}^n c_k^2}  \right)
\]

We split the proof of Theorem~\ref{thm:termination} into four propositions
 (namely Propositions \ref{prop:term-pumping}--\ref{prop:term-lower-critical} below),
 which together imply Theorem~\ref{thm:termination}.

The following proposition establishes the lower bound from Theorem~\ref{thm:termination}~(2):
\newcommand{\stmtproppumping}{
 Let $\Delta$ be an almost surely terminating pBPA with stack alphabet~$\Gamma$.
 Let $\pmin =\min\{ p \mid X\btran{p}\alpha\text{ in }\Delta\}$.
 Assume that $\tailprob{X_0}{2^{|\Gamma|}} > 0$.
 Then we have
 \[
  \pmin^{n} \quad \le \quad \tailprob{X_0}{n} \qquad \text{for all $n \in \Nset$.}
 \]
}
\begin{proposition} \label{prop:term-pumping}
 \stmtproppumping
\end{proposition}
\begin{proof}
 Let $\termt{X_0}(w) \ge n$ for some $n \in \Nset$ and some $w \in \run(X_0)$.
 It follows from the definition of the probability space of a pPDA that the set of all runs starting with $w(0), w(1), \ldots, w(n)$
  has a probability of at least $\pmin^n$.
 Therefore, in order to complete the proof, it suffices to show that $\tailprob{X_0}{2^{|\Gamma|}} > 0$ implies
   $\tailprob{X_0}{n} > 0$ for all $n \in \Nset$.

 To this end, we use a form of the pumping lemma for context-free languages.
 Notice that a pBPA can be regarded as a context-free grammar with probabilities (a stochastic context-free grammar)
  with an empty set of terminal symbols and $\Gamma$ as the set of nonterminal symbols.
 Each finite run $w \in \run(X_0)$ corresponds to a derivation tree with root $X_0$ that derives the word~$\varepsilon$.
 The termination time $\termt{X_0}$ is the number of (internal) nodes in the tree.
 In the rest of the proof we use this correspondence.

 Let $\tailprob{X_0}{2^{|\Gamma|}} > 0$.
 Then there is a run $w \in \run(X_0)$ with $\termt{X_0}(w) \ge 2^{|\Gamma|}$.
 This run~$w$ corresponds to a derivation tree with at least $2^{|\Gamma|}$ (internal) nodes.
 In this tree there is a path from the root (labeled with~$X_0$) to a leaf such that on this path
  there are two different nodes, both labeled with the same symbol.
 Let us call those nodes $n_1$ and $n_2$, where $n_1$ is the node closer to the root.
 By replacing the subtree rooted at~$n_2$ with the subtree rooted at~$n_1$ we obtain a larger derivation tree.
 This completes the proof.
\qed
\end{proof}

The following proposition establishes the upper bound of Theorem~\ref{thm:termination}~(2):
\begin{proposition} \label{prop:term-upper-subcritical}
 Let $\Delta$ be an almost surely terminating pBPA with stack alphabet~$\Gamma$.
 Assume that $X_0$ depends on all $X \in \Gamma \setminus \{X_0\}$.
 Define
 \[
  \Emax := \max_{X \in \Gamma} E[X] \qquad \text{and} \qquad
    B := \max_{X \btran{} \alpha} \left| 1 - E[X] + \sum_{Y \in \Gamma} \numbsymb{Y}(\alpha) \cdot E[Y] \right|\,.
 \]
 Then for all $n\in \Nset$ with $n \ge 2 E[X_0]$ we have
 \[
  \tailprob{X_0}{n} \qquad\leq\qquad \exp \frac{2 E[X_0] -n}{2B^2} \quad\leq\quad \exp\left( 1 - \frac{n}{8 \Emax^2} \right) \,.
\]
\end{proposition}
\begin{proof}
 Let $w\in \run(X_0)$.
 We denote by $I(w)$ the maximal number $j\geq 0$ such that $w(j-1)\not =\varepsilon$.
 Given $i\geq 0$, we define $\ms{i}(w) := E[w(i)] + \min\{i,I(w)\}$.
 We prove that $E(\ms{i+1}\mid \ms{i})=\ms{i}$, i.e., $\ms{0},\ms{1},\ldots$ forms a martingale.
 It has been shown in~\cite{EKM:prob-PDA-expectations} that
 \begin{eqnarray*}
 E[X] & = & \sum_{X\btran{x} \varepsilon} x + \sum_{X\btran{x} Y} x\cdot (1 + E[Y]) +
                  \sum_{X\btran{x} YZ} x\cdot (1 + E[Y] + E[Z]) \\
                  & = & 1+\sum_{X\btran{x} Y} x\cdot E[Y] +
                  \sum_{X\btran{x} YZ} x\cdot (E[Y] + E[Z])\,. \\
 \end{eqnarray*}
 On the other hand, let us fix a path $u\in \fpath(X_0)$ of length $i$ and let $w$ be an arbitrary run of $\run(u)$.
 First assume that $u(i-1)=X\alpha\in \Gamma\Gamma^*$.
 Then we have:
 \begin{align*}
  & \Ex{\ms{i+1}\mid \run(u)} \\
  & = \sum_{X\btran{x} \varepsilon} x\cdot (\ms{i}(w)-E[X]+1) +
      \sum_{X\btran{x} Y} x\cdot (\ms{i}(w)-E[X]+E[Y]+1) + \\
  & \quad + \sum_{X\btran{x} YZ} x\cdot (\ms{i}(w)-E[X] + E[Y] + E[Z]+1) \\
  & = \ms{i}(w) - E[X] + 1 + \sum_{X\btran{x} Y} x\cdot E[Y] +
                   \sum_{X\btran{x} YZ} x\cdot (E[Y] + E[Z]) \\
  & = \ms{i}(w)
 \end{align*}
 If $u(i-1)=\varepsilon$, then for every $w\in \run(u)$ we have $\ms{i+1}(w)=I(w)=\ms{i}(w)$.
 This proves that $\ms{0},\ms{1},\ldots$ is a martingale.

 By Azuma's inequality (see~\cite{book:Williams}), we have
 \begin{align*}
 \calP(\ms{n} - E[X_0]\geq n - E[X_0])
  & \quad \leq \quad \exp\left(\frac{-(n-E[X_0])^2}{2\sum_{k=1}^n B^2}\right)
    \quad \le \quad \exp\left(\frac{2 E[X_0] -n}{2B^2}\right)\,.
 \end{align*}
 For every $w\in \run(X_0)$ we have that $w(n)\not = \varepsilon$ implies $\ms{n}\geq n$.
 It follows:
 \[
 \tailprob{X_0}{n} \quad \leq \quad \calP(\ms{n}\geq n) \quad \leq \quad \exp \left(\frac{2 E[X_0] -n}{2B^2}\right)
                   \quad \le \quad \exp\left( 1 - \frac{n}{8 \Emax^2} \right)\,,
 \]
 where the final inequality follows from the inequality $B \le 2 \Emax$.
\qed
\end{proof}

The following proposition establishes the upper bound of Theorem~\ref{thm:termination}~(3):
\newcommand{\stmtproptermuppercritical}
{
 Let $\Delta$ be an almost surely terminating pBPA with stack alphabet~$\Gamma$.
 Assume that $X_0$ depends on all $X \in \Gamma \setminus \{X_0\}$.
 Let $\pmin =\min\{ p \mid X\btran{p}\alpha\text{ in }\Delta\}$.
 Let $h$ denote the height of the DAG of SCCs.
 Then there is $n_0 \in \Nset$ such that
 \[
  \tailprob{X_0}{n} \quad \le \quad \frac{18 h |\Gamma| / \pmin^{3|\Gamma|}}{n^{1/(2^{h+1}-2)}} \qquad \text{for all $n \ge n_0$.}
 \]
}
\begin{proposition} \label{prop:term-upper-critical}
 \stmtproptermuppercritical
\end{proposition}
\begin{proof}[sketch; a full proof is given in Section~\ref{app:term-upper-critical}]
Assume that $E[X_0]$ is infinite.
To give some idea of the (quite involved) proof, let us first consider
a simple pBPA~$\Delta$ with $\Gamma = \{X\}$ and the rules $X \btran{1/2} XX$
and $X \btran{1/2} \varepsilon$.  In fact, $\Delta$ is closely related
to a simple random walk starting at~$1$, for which the time until it
hits~$0$ can be exactly analyzed (see, e.g.,~\cite{book:Williams}).
Clearly, we have $h=|\Gamma|=1$ and $\pmin=1/2$.
Theorem~\ref{thm:termination}(3)
implies
$\tailprob{X}{n} \in \bigo(1/\sqrt{n})$.  Let us sketch why
this upper bound holds.

\newcommand{\mos}[1]{\widetilde{m}^{(#1)}}
Let $\theta > 0$, define $g(\theta) := \frac12 \cdot \exp(-\theta
\cdot (-1)) + \frac12 \cdot \exp(-\theta \cdot (+1))$, and define for
a run $w \in \run(X)$ the sequence
 \[
  \ms{i}_\theta(w) =
   \begin{cases}
           \exp(-\theta \cdot |w(i)|) / g(\theta)^i  & \text{if $i=0$ or $w(i-1) \ne \varepsilon$} \\
           \ms{i-1}_\theta(w)                        & \text{otherwise.}
   \end{cases}
 \]
 One can show (cf.~\cite{book:Williams}) that $\ms{0}_\theta,
 \ms{1}_\theta, \ldots$ is a martingale, i.e.,
 $\Ex{\ms{i}_\theta \mid \ms{i-1}_\theta} = \ms{i-1}_\theta$ for all
 $\theta > 0$.  Our proof crucially depends on some analytic
 properties of the function $g: \Rset \to \Rset$: It is easy to verify
 that $1 = g(0) < g(\theta)$ for all $\theta > 0$, and $0 = g'(0)$,
 and $1 = g''(0)$.
 One can show that Doob's Optional-Stopping Theorem (see
 Theorem~10.10~(ii) of~\cite{book:Williams}) applies, which implies
 $\ms{0}_\theta = \Ex{\ms{\termt{X}}_\theta}$.  It follows that for all
 $n \in \Nset$ and $\theta > 0$ we have that
\begin{align}
 \exp(-\theta) & = \ms{0}_\theta  \ =\ \Ex{\ms{\termt{X}}_\theta} \ =\ \Ex{g(\theta)^{-\termt{X}}} \ =\
      \sum_{i=0}^\infty \calP(\termt{X} = i) \cdot g(\theta)^{-i} \label{eq:first-line} \\
 & \le\  \sum_{i=0}^{n-1} \calP(\termt{X}=i) \cdot 1 + \sum_{i=n}^\infty \calP(\termt{X} = i) \cdot g(\theta)^{-n} \nonumber \\
 & =\ 1 - \calP(\termt{X} \ge n) + \calP(\termt{X} \ge n) \cdot g(\theta)^{-n} \nonumber
\end{align}
Rearranging this inequality yields $\calP(\termt{X} \ge n) \le \frac{1 - \exp(-\theta)}{1 - g(\theta)^{-n}}$,
from which one obtains, setting $\theta := 1/\sqrt{n}$,
and using the mentioned properties of~$g$ and several applications of l'Hopital's rule,
that \mbox{$\calP(\termt{X} \ge n) \in \bigo(1/\sqrt{n})$}.

Next we sketch how we generalize this proof to pBPA that consist of
only one SCC, but have more than one stack symbol.  In this case, the
term $|w(i)|$ in the definition of~$\ms{i}_\theta(w)$ needs to be
replaced by the sum of {\em weights} of the symbols in~$w(i)$.  Each
$Y \in \Gamma$ has a weight which is drawn from the dominant
eigenvector of a certain matrix, which is characteristic for~$\Delta$.
Perron-Frobenius theory guarantees the existence of a suitable weight
vector~$\vu \in \Rset_+^\Gamma$.
The function $g$ consequently needs to be replaced by a function $g_Y$
for each $Y \in \Gamma$.  We need to keep the property that $g_Y''(0)
> 0$.  Intuitively, this means that $\Delta$ must have, for each $Y
\in \Gamma$, a rule $Y \btran{} \alpha$ such that $Y$ and $\alpha$
have different weights.  This can be accomplished by
transforming~$\Delta$ into a certain normal form.

Finally, we sketch how the proof is generalized to pBPA with more than
one SCC.  For simplicity, assume that $\Delta$ has only two stack
symbols, say $X$ and $Y$, where $X$ depends on~$Y$, but $Y$ does not
depend on~$X$.  Let us change the execution order of pBPA as follows:
whenever a rule with $\alpha \in \Gamma^*$ on the right hand side
fires, then all $X$-symbols in~$\alpha$ are added on top of the stack,
but all $Y$-symbols are added at the {\em bottom} of the stack.  This
change does not influence the termination time of pBPA, but it allows
to decompose runs into two phases: an $X$-phase where $X$-rules are
executed which may produce \mbox{$Y$-symbols} or further $X$-symbols;
and a $Y$-phase where $Y$-rules are executed which may produce further
$Y$-symbols but no $X$-symbols, because $Y$ does not depend on~$X$.
Arguing only qualitatively, assume that $\termt{X}$ is ``large''.
Then either (a) the $X$-phase is ``long'' or (b) the $X$-phase is
``short'', but the $Y$-phase is ``long''.  For the probability of
event~(a) one can give an upper bound using the bound for one SCC,
because the produced $Y$-symbols can be ignored.  For event~(b),
observe that if the $X$-phase is short, then only few $Y$-symbols can
be created during the $X$-phase.  For a bound on the probability of
event~(b) we need a bound on the probability that a pBPA with one SCC
and a ``short'' initial configuration takes a ``long'' time to
terminate.  The previously sketched proof for an initial configuration
with a single stack symbol can be suitably generalized to handle other
``short'' configurations. All details are given in Section~\ref{app:term-upper-critical}.
\qed
\end{proof}

The following proposition establishes the lower bound of Theorem~\ref{thm:termination}~(3):
\newcommand{\stmtproptermlowercritical}
{
 Let $\Delta$ be an almost surely terminating pBPA with stack alphabet~$\Gamma$.
 Assume that $X_0$ depends on all $X \in \Gamma \setminus \{X_0\}$.
 Assume $E[X_0] = \infty$.
 Then there is $c > 0$ such that
 \[
  \frac{c}{\sqrt{n}} \quad \le \quad \tailprob{X_0}{n} \qquad \text{for all $n \in \Nset$.}
 \]
}
\begin{proposition} \label{prop:term-lower-critical}
\stmtproptermlowercritical
\end{proposition}
The proof of Proposition~\ref{prop:term-lower-critical} follows the lines of the previous proof sketch, but with an additional trick:
To obtain the desired bound, one needs to take the derivative with respect to~$\theta$ on both sides of Equation~\eqref{eq:first-line}.
The full proof is given in Section~\ref{app:term-lower-critical}.

\medskip
\noindent
\textbf{Tightness of the bounds in the case of infinite expectation.}
If $E[X_0]$ is infinite, the lower and upper bounds of Theorem~\ref{thm:termination}~(3) asymptotically coincide in the ``strongly connected'' case
 (i.e., where $h=1$ holds for the height of the DAG of the SCCs of the dependence relation).
In other words, in the strongly connected case we must have $\calP(\termt{} \ge n) \in \Theta(1/\sqrt{n})$.
Otherwise (i.e., for larger~$h$) the upper bound in Theorem~\ref{thm:termination}~(3) cannot be substantially tightened.
This follows from the following proposition:
\newcommand{\stmtpropheavytail}{
 Let $\Delta_h$ be the pBPA with $\Gamma_h = \{X_1, \ldots, X_h\}$
 and the following rules:
 {
  \[
   X_h \bbtran{1/2} X_h X_h\,,\, X_h \bbtran{1/2} X_{h-1}\,,\,\ldots\,,\,
   X_2 \bbtran{1/2} X_2 X_2\,,\, X_2 \bbtran{1/2} X_1\,,\;
   X_1 \bbtran{1/2} X_1 X_1\,,\, X_1 \bbtran{1/2} \varepsilon\]}%
Then $[X_h] = 1$, $E[X_h] = \infty$, and there is $c_h > 0$ with
 \[
  \frac{c_h}{n^{1/2^h}} \quad\le\quad \tailprob{X_h}{n} \qquad \text{for all $n \in \Nset$}.
 \]
}
\begin{proposition} \label{prop:heavy-tail}
 \stmtpropheavytail
\end{proposition}
Proposition~\ref{prop:heavy-tail} is proved in Section~\ref{app-optimal}.


\section{Conclusions and Future Work}
\label{sec:concl}

We have provided a reduction from stateful to stateless pPDA which gives new
insights into the theory of pPDA and at the same time simplifies it
substantially. We have used this reduction and martingale theory to exhibit
a dichotomy result that precisely characterizes the distribution of the
termination time in terms of its expected value.

Although the bounds presented in this paper are asymptotically optimal,
there is still space for improvements.
%
%
We conjecture that our results can be extended to more general
reward-based models, where each configuration is assigned a nonnegative
reward and the total reward accumulated in a given service is considered
instead of its length. This is particularly challenging if the rewards
are unbounded (for example, the reward assigned to a given
configuration may correspond to the total memory allocated by the
procedures in the current call stack). Full answers to these questions
would generalize some of the existing deep results about simpler models, and
probably reveal an even richer underlying theory of pPDA which is still
undiscovered.


\section{Proofs} \label{sec:proofs}

In this section we give the missing proofs for the stated results.
Some additional notation is used in the proofs.

\stefan{I'm not aware that this is used except possibly for the transformation part. If the notation is used, it should probably be defined elsewhere.}
\begin{itemize}
\item
 Given two sets $K \subseteq \Sigma^*$ and
 $L \subseteq \Sigma^* \cup \Sigma^\omega$, we use $K \cdot L$ (or just
 $KL$) to denote the concatenation of $K$ and $L$, i.e.,
 $KL = \{ww' \mid w \in K, w'\in L \}$.
\item For a run $w$ and $i \in \Nset$, we write $w_i$ to denote
  the run $w(i)\,w(i{+}1) \dots$.
\end{itemize}

\subsection{Proofs of Propositions~\ref{lem:longrun_PDA_to_BPA}
and~\ref{prop:pdatobpa}}

\begin{refproposition}{lem:longrun_PDA_to_BPA}
Let $p_0X_0 \in Q \times \Gamma$ such that $[p_0X_0{\uparrow}] = 1$. Then
there is a partial function
$\Upsilon:\run[M_{\Delta}](p_0X_0)\rightarrow
\run[M_{\Delta_2}](\co{p_0X_0{\uparrow}})$ such that for every
$w \in \run[M_{\Delta}](p_0X_0)$, where $\Upsilon(w)$ is defined, and every
$n \in \Nset$ we have the following: if $w(n) = qY\beta$, then
$\Upsilon(w)(n) = \co{qY{\dag}}\gamma$, where $\dag$ is either an element
of $Q$ or ${\uparrow}$. Further, for every measurable set of runs
$R \subseteq \run[M_{\Delta_2}](\co{p_0X_0{\uparrow}})$ we have that
$\Upsilon^{-1}(R)$ is measurable and $\calP(R)=\calP(\Upsilon^{-1}(R))$.
\end{refproposition}
\begin{proof}
Let $w \in \run[M_\Delta](p_0X_0)$. We define an infinite sequence $\bar{w}$
over $\bar{\Gamma}^*$ inductively as follows:
\begin{itemize}
\item $\bar{w}(0) = \co{p_0X_0{\uparrow}}$
\item If $\bar{w}(i) = \varepsilon$ (which intuitively means that
  an ``error'' was indicated while defining the first $i$ symbols of~$w$),
  then $w(i{+}1) = \varepsilon$. Now let us assume that
  $\bar{w}(i) = \co{pX\dag} \alpha$, where
  $\dag \in Q \cup \{{\uparrow}\}$, and $w(i) = pX\gamma$ for some
  $\gamma \in \Gamma^*$. Let $pX \btran{} r\beta$ be the rule of
  $\Delta$ used to derive the transition
  $w(i) \tran{} w(i{+}1)$.  Then
  \[
    \bar{w}(i{+}1) =
    \begin{cases}
      \alpha &
         \text{if $\beta = \varepsilon$ and $\dag = r$;}\\[1ex]
      \co{rY\dag}\alpha &
         \text{if $\beta = Y$ and $[rY\dag] > 0$;}\\[1ex]
      \co{rYs}\co{sZ\dag}\alpha &
         \text{if $\beta = YZ$, $[sZ\dag] > 0$, and there is
         $k > i$ such that $w(k) = sZ\gamma$ and}\\
         & |w(j)| > |w(i)| \text{ for all } i <j <k;\\[1ex]
      \co{rY{\uparrow}}\alpha &
         \text{if $\beta = YZ$, $[rY{\uparrow}] > 0$, and
          $|w(j)| > |w(i)|$ for all $j > i$;}\\[1ex]
      \varepsilon & \text{otherwise.}
    \end{cases}
  \]
\end{itemize}
We say that $w \in \run[M_\Delta](p_0X_0)$ is \emph{valid} if
$\bar{w}(i) \neq \varepsilon$ for all $i \in \Nset$. One
can easily check that if $w$ is valid, then $\bar{w}$ is a run of
$\bar{\Delta}$ initiated in $\co{p_0X_0{\uparrow}}$. We put
$\Upsilon(w) = \bar{w}$ for all valid $w \in \run[M_\Delta](p_0X_0)$.
For invalid runs, $\Upsilon$ stays undefined.

It follows directly from the definition of $\bar{w}$ that for every
valid $w \in \run[M_{\Delta}](p_0 X_0)$ and every $i \in \Nset$ we have
that if $w(i) = qY\beta$ then $\bar{w}(i) = \co{qY{\dag}}\gamma$, where
$\dag \in Q \cup \{{\uparrow}\}$.

Now we check that for every measurable set of runs
$R \subseteq \run[M_{\bar{\Delta}}](\co{p_0X_0{\uparrow}})$ we have that
$\Upsilon^{-1}(R)$ is measurable and $\calP(R)=\calP(\Upsilon^{-1}(R))$.
First, realize that the set of all \emph{invalid}
$w \in \run[M_{\Delta}](p_0 X_0)$ is measurable and its probability
is zero. Hence, it suffices to show that for every finite path
$\bar{v}$ in $M_{\bar{\Delta}}$ initiated in $\co{p_0X_0{\uparrow}}$ we have
that $\Upsilon^{-1}(\run[M_{\bar{\Delta}}](\bar{v}))$ is measurable
and $\calP(\Upsilon^{-1}(\run[M_{\bar{\Delta}}](\bar{v}))) =
\calP(\run[M_{\bar{\Delta}}](\bar{v}))$. For simplicity, we write just
$\Upsilon^{-1}(\bar{v})$ instead of
$\Upsilon^{-1}(\run[M_{\bar{\Delta}}](\bar{v}))$.

Observe that every  configuration $\bar{\gamma}$ reachable
from $\co{p_0X_0{\uparrow}}$ in $M_{\bar{\Delta}}$ is of the form
$\bar{\gamma} = \co{p_1X_1p_2}\cdots\co{p_kX_kp_{k+1}}\co{p_{k+1}Y{\uparrow}}$
where $k \geq 0$. We put
\[
   P[\bar{\gamma}] \quad = \quad
   [p_1X_1p_2]\cdots [p_kX_kp_{k+1}] \cdot [p_{k+1}Y{\uparrow}]
\]
Further, we say that a configuration $p\alpha$ of $\Delta$ is \emph{compatible}
with $\bar{\gamma}$ if $p = p_1$ and $\alpha = X_1\cdots X_kY \beta$
for some $\beta \in \Gamma^*$. A run $w$ initiated
in such a compatible configuration $p_1 X_1\cdots X_kY \beta$ \emph{models}
$\bar{\gamma}$, written $w \models \bar{\gamma}$, if $w$ is of the form
\[
 p_1 X_1\cdots X_kY \beta \quad \tran{}^* \quad p_2X_2\cdots X_kY \beta
 \quad \tran{}^* \quad
 \cdots \quad \tran{}^* p_{k+1}Y\beta \quad \tran{} \quad \cdots
\]
where for all $1 \leq i \leq k$, the stack length of all intermediate
configurations visited along the subpath
$p_i X_i\cdots X_kY \beta \tran{}^* p_{i+1} X_{i+1}\cdots X_kY \beta$
is at least $|X_i\cdots X_kY\beta|$. Further, the stack length in all
configurations visited after $q_kY\beta$ is at least $|Y\beta|$.
A straightforward induction on $k$ reveals that
\begin{equation}
 \calP\left\{w \in \run(p_1 X_1\cdots X_kY \beta) \mid
     w \models \bar{\gamma}\right\}
 \quad = \quad P[\bar{\gamma}]
\label{eq-Pgamma}
\end{equation}
Let $\bar{v} \bar{\alpha}$, where $\bar{\alpha} \in \bar{\Gamma}^*$,
be a finite path in $M_{\bar{\Delta}}$ initiated in
$\co{p_0X_0{\uparrow}}$, and let $\calE(\bar{v}\bar{\alpha})$ be the set of all
finite path $v A$ in $M_\Delta$ initiated in $p_0X_0$ such that
$A \in Q\times\Gamma^*$,
$|vA| = |\bar{v}\bar{\alpha}|$, and $\Upsilon^{-1}(\bar{v} \bar{\alpha})$
contains a run that starts with $vA$.
One can easily check that if $vA \in \calE(\bar{v}\bar{\alpha})$,
then $A$ is compatible with $\bar{\alpha}$.
Further,
\begin{equation}
  \Upsilon^{-1}(\bar{v}\bar{\alpha}) = \bigcup_{v A \in \calE(\bar{v}\bar{\alpha})}
     vA \odot \big\{w \in \run[M_\Delta](A) \mid w \models
     \bar{\alpha} \big\}
\label{eq-Upsilon}
\end{equation}
From (\ref{eq-Upsilon}) we obtain that $\Upsilon^{-1}(\bar{v}\bar{\alpha})$
is measurable, and by combining (\ref{eq-Pgamma}) and (\ref{eq-Upsilon})
we obtain
\begin{equation}
  \calP(\Upsilon^{-1}(\bar{v}\bar{\alpha})) \quad = \quad
  P[\bar{\alpha}] \cdot \sum_{v A \in \calE(\bar{v}\bar{\alpha})} \calP(\run(vA))
\label{eq-Upsilon-prob}
\end{equation}
Now we show that $\calP(\Upsilon^{-1}(\bar{v}\bar{\alpha})) =
\calP(\run(\bar{v}\bar{\alpha}))$.
We proceed by induction on $|\bar{v}\bar{\alpha}|$. The base case when
$\bar{v}\bar{\alpha} = \co{p_0X_0{\uparrow}}$ is immediate.
Now suppose that $\bar{v}\bar{\alpha} = \bar{u} \bar{\beta} \bar{\alpha}$,
where $\bar{\beta} \tran{x} \bar{\alpha}$. By applying
(\ref{eq-Upsilon}) and (\ref{eq-Upsilon-prob}) we obtain
\[
\begin{array}{lclr}
  \calP(\Upsilon^{-1}(\bar{u} \bar{\beta} \bar{\alpha}))
  & = & \displaystyle
  \calP\left(\bigcup_{u B A \in \calE(\bar{u} \bar{\beta} \bar{\alpha})}
     u\,B\,A \odot \left\{w \in \run(A) \mid w
     \models \bar{\alpha}\right\}\right)\\[2em]
  & = & \displaystyle
  \calP\left(\bigcup_{u B \in \calE(\bar{u} \bar{\beta})}
     u\,B \odot \bigcup_{A \in Q \times \Gamma^*}
     \left\{w \in \run(BA) \mid uBA \in
         \calE(\bar{u} \bar{\beta} \bar{\alpha}), w \models
     \bar{\beta}, w_1 \models
     \bar{\alpha} \right\}\right)\\[2em]
  & = & \displaystyle
  \sum_{u B \in \calE(\bar{u} \bar{\beta})}
     \calP(\run(u\,B)) \cdot \calP\left(\bigcup_{A \in Q \times \Gamma^*}
     \left\{w \in \run(BA) \mid uBA \in
         \calE(\bar{u} \bar{\beta} \bar{\alpha}), w \models
     \bar{\beta}, w_1 \models
     \bar{\alpha} \right\}\right)\\[2em]
  & =^* & \displaystyle
  \sum_{u B \in \calE(\bar{u} \bar{\beta})}
     \calP(\run(u\,B)) \cdot P[\bar{\beta}] \cdot x\\[2em]
  & = & \displaystyle
  x \cdot \calP(\Upsilon^{-1}(\bar{u}\bar{\beta}))\\[2em]
  & = & \displaystyle
  \calP(\run(\bar{u}\bar{\beta}\bar{\alpha}))
\end{array}
\]
The (*) equality is proved by case analysis (we distinguish
possible forms of the rule which generates the transition
$\bar{\beta} \tran{x} \bar{\alpha}$).
\qed
\end{proof}

\begin{refproposition}{prop:pdatobpa}
  Let $pXq\in Q\times \Gamma\times Q$ and $[pXq]>0$.  Then almost all
  runs of $M_{\bpadelta}$ initiated in $\langle pXq\rangle$ terminate, i.e., reach
  $\varepsilon$.  Further, for all $n\in \Nset$ we have that
  \[
    \massprobcond{pX}{n}{\run(pXq)} \quad = \quad
    \massprobcond{\langle pXq\rangle}{n}{\run(\langle pXq \rangle)}
  \]
\end{refproposition}
\begin{proof}
For every $n\in \Nset$ we define
\begin{eqnarray*}
   D_{pXq}(n) & := &
    \calP(\run(pXq), \ \termt{pX} = n \mid \run(pX))\\
   D_{\langle pXq\rangle}(n) & := &
    \calP(\termt{\langle pXq\rangle} = n \mid \run(\langle pXq \rangle))
\end{eqnarray*}
We prove the following:
\begin{equation}
 D_{pXq}(n) = [pXq] \cdot D_{\langle pXq\rangle}(n)\,. \label{eq:pdatobpa-induction}
\end{equation}
Notice that~\eqref{eq:pdatobpa-induction} implies
 $\massprobcond{pX}{n}{\run(pXq)}  = \massprobcond{\langle pXq\rangle}{n}{\run(\langle pXq \rangle)}$,
 as $\massprobcond{pX}{n}{\run(pXq)} = D_{pXq}(n) / [pXq]$.

To prove~\eqref{eq:pdatobpa-induction}, we proceed by induction on~$n$.
First, assume that $n=1$. If $pX\btran{x} q\varepsilon$,
then $\langle pXq \rangle \btran{y}\varepsilon$, where $y=\frac{x}{[pXq]}$ and thus
\[
D_{pXq}(1)=x=\frac{[pXq] x}{[pXq]}=[pXq]y=[pXq] D_{\langle pXq\rangle}(1)\,.
\]
If there is no rule $pX\btran{} q\varepsilon$ in $\Delta$, then there is no rule $\langle pXq\rangle\btran{} \varepsilon$ in $\bpadelta$.

Assume that $n>1$. Let us first prove that $D_{pXq}(n)$ can be decomposed according to the first step:
\begin{equation}\label{eq:pPDA-decomp}
D_{pXq}(n) = \sum_{pX\btran{x} rY} x\cdot D_{rYq}(n-1) + \sum_{i=1}^{n-1}\,
\sum_{pX\btran{x} rYZ}\, \sum_{s\in Q} x\cdot D_{rYs}(i)\cdot D_{sZq}(n-i-1)
\end{equation}
To prove~(\ref{eq:pPDA-decomp}) we introduce some notation. For every $rYs\in Q\times \Gamma\times Q$ and
$i\in \Nset$ we denote by $B_{rYs}(i)$ the set of all paths from $rY$ to $s\varepsilon$ of length $i$.
We also denote by $B_{rYs}(i)\lfloor Z$ the set of all paths of the form $p_0 \alpha_0 Z\cdots p_i\alpha_i Z$
where $p_0 \alpha_0 \cdots p_i\alpha_i$ belongs to $B_{rYs}(i)$.
We have
\[
B_{pXq}(n)=\bigcup_{pX\btran{} rY} B_{rYs}(n-1) \cup
\bigcup_{i=1}^{n-1}\,
\bigcup_{pX\btran{x} rYZ}\, \bigcup_{s\in Q} \{pX\}\cdot B_{rYs}(i)\lfloor Z\cdot B_{sZq}(n-i-1)
\]
where all the unions are disjoint.
Now the probability of following a path of $B_{rYs}(i)\lfloor Z$ is equal to the probability of following
a path of $B_{rYs}(i)$, which is $D_{rYs}(i)$. Thus we have that
\begin{eqnarray*}
\calP(\run(\{pX\}\cdot B_{rYs}(i)\lfloor Z\cdot B_{sZq}(n-i-1))) & = &
         x\cdot \calP(B_{rYs}(i)\lfloor Z\cdot \run(B_{sZq}(n-i-1))) \\
         & = & x\cdot \calP(\run(B_{rYs}(i))\lfloor Z)\cdot \calP(\run(B_{sZq}(n-i-1))) \\
         & = & x\cdot \calP(\run(B_{rYs}(i)))\cdot D_{sZq}(n-i-1) \\
         & = & x \cdot D_{rYs}(i) \cdot D_{sZq}(n-i-1)\,.
\end{eqnarray*}
It follows that
\begin{eqnarray*}\label{eq:pPDA-decomp-paths}
D_{pXq}(n) & = & \calP(\run(B_{pXq}(n))) \\
           & = & \calP(\run\left(\bigcup_{pX\btran{} rY} B_{rYs}(n-1) \cup
\bigcup_{i=1}^{n-1}\,
\bigcup_{pX\btran{x} rYZ}\, \bigcup_{s\in Q} \{pX\}\cdot B_{rYs}(i)\lfloor Z\cdot B_{sZq}(n-i-1)\right)) \\
           & = &  \sum_{pX\btran{x} rY} x\cdot \calP(\run(B_{rYs}(n-1))) + \\
           &   &  \quad + \sum_{i=1}^{n-1}\,
\sum_{pX\btran{x} rYZ}\, \sum_{s\in Q} x\cdot \calP(\run(B_{rYs}(i)))\cdot
\calP(\run(B_{sZq}(n-i-1))) \\
           & = & \sum_{pX\btran{x} rY} x\cdot D_{rYq}(n-1) + \sum_{i=1}^{n-1}\,
\sum_{pX\btran{x} rYZ}\, \sum_{s\in Q} x\cdot D_{rYs}(i)\cdot D_{sZq}(n-i-1)\,,
\end{eqnarray*}
which proves~\eqref{eq:pPDA-decomp}.
Now we are ready to finish the induction proof of~\eqref{eq:pdatobpa-induction}.
{\small\begin{eqnarray*}
D_{pXq}(n) & = & \sum_{pX\btran{x} rY} x\cdot D_{rYq}(n-1) + \sum_{i=1}^{n-1}\,
\sum_{pX\btran{x} rYZ}\, \sum_{s\in Q} x\cdot D_{rYs}(i)\cdot D_{sZq}(n-i-1) \\
           & = & \sum_{pX\btran{x} rY} x\cdot D_{\langle rYq\rangle}(n-1)\cdot [rYq] + \\
           &   & \quad +\, \sum_{i=1}^{n-1}\,
                 \sum_{pX\btran{x} rYZ}\, \sum_{s\in Q} x\cdot D_{\langle rYs\rangle}(i)\cdot [rYs]\cdot
                 D_{\langle sZq \rangle}(n-i-1)\cdot [sZq] \\
           & = & [pXq]\cdot \left(\sum_{pX\btran{x} rY} \frac{x[rYq]}{[pXq]}\cdot
                 D_{\langle rYq\rangle}(n-1) + \right. \\
           &   & \quad +\, \left. \sum_{i=1}^{n-1}\,
                 \sum_{pX\btran{x} rYZ}\, \sum_{s\in Q} \frac{x[rYs][sZq]}{[pXq]}\cdot
                 D_{\langle rYs\rangle}(i)\cdot
                 D_{\langle sZq \rangle}(n-i-1)\right) \\
           & = & [pXq]\cdot \left(\sum_{\langle pXq\rangle\btran{y} \langle rYq\rangle} y\cdot
                 D_{\langle rYq\rangle}(n-1) + \right. \\
           &   & \quad +\, \left. \sum_{i=1}^{n-1}\,
                 \sum_{\langle pXq\rangle\btran{y} \langle rYs\rangle\langle sZq\rangle} y\cdot
                 D_{\langle rYs\rangle}(i)\cdot
                 D_{\langle sZq \rangle}(n-i-1)\right) \\
           & = & [pXq]\cdot D_{\langle pXq\rangle}(n)
\end{eqnarray*}}%
Finally, observe that $\sum_{n=1}^{\infty} D_{\langle pXq\rangle}$ is the probability of reaching
$\varepsilon$ from $\langle pXq\rangle$ and that
\[
\sum_{n=1}^{\infty} D_{\langle pXq\rangle}=\sum_{n=1}^{\infty} \frac{D_{pXq}(n)}{[pXq]}=\frac{1}{[pXq]}\cdot
\sum_{n=1}^{\infty} D_{pXq}(n)=1\,.
\]
\qed
\end{proof}


\subsection{Proof of Proposition~\ref{prop:term-upper-critical}} \label{app:term-upper-critical}

In this subsection we prove Proposition~\ref{prop:term-upper-critical}.
Given a finite set~$\Gamma$, we regard the elements of $\Rset^\Gamma$ as vectors.
Given two vectors $\vu, \vv \in \Rset^\Gamma$, we define a scalar product by setting $\vu \thickdot \vv := \sum_{X \in \Gamma} \vu(X) \cdot \vv(X)$.
Further, elements of~$\Rset^{\Gamma \times \Gamma}$ are regarded as matrices, with the usual matrix-vector multiplication.

It will be convenient for the proof to measure the termination time of pBPA starting in an arbitrary initial configuration~$\alpha_0 \in \Gamma\Gamma^*$,
 not just with a single initial symbol $X_0 \in \Gamma$.
To this end we generalize $\termt{X_0}$, $\run(X_0)$, etc.\ to $\termt{\alpha_0}$, $\run(\alpha_0)$, etc.\ in the straightforward way.

It will also be convenient to allow ``pBPA'' that have transition rules with more than two stack symbols on the right-hand side.
We call them {\em relaxed pBPA}.
All concepts associated to a pBPA, e.g., the induced Markov chain, termination time, etc., are defined analogously for relaxed pBPA.

A relaxed pBPA is called {\em strongly connected}, if the DAG of the dependence relation on its stack alphabet consists of a single SCC.

For any $\alpha \in \Gamma^*$, define $\numbsymb{\alpha}$ as the Parikh image of~$\alpha$,
 i.e., the vector of $\Nset^{\Gamma}$ such that $\numbsymb{\alpha}(Y)$ is the number of occurrences of $Y$ in $\alpha$.
Given a relaxed pBPA~$\Delta$, let $A_\Delta \in \Rset^{\Gamma\times\Gamma}$ be the matrix with
 \[
  A_\Delta(X,Y) = \sum_{X \btran{p} \alpha} p \cdot \numbsymb{\alpha}(Y)\,.
 \]
We drop the subscript of~$A_\Delta$ if $\Delta$ is clear from the context.
Intuitively, $A(X,Y)$ is the expected number of $Y$-symbols pushed on the stack when executing a rule with $X$ on the left hand side.
For instance, if $X \btran{1/5} XX$ and $X \btran{4/5} \varepsilon$, then $A(X,X) = 2/5$.
Note that $A$ is nonnegative.
The matrix~$A$ plays a crucial role in the analysis of pPDA and related models (see e.g.~\cite{EY:RMC-SG-equations-JACM})
 and in the theory of branching processes~\cite{Harris:book}.
We have the following lemma:
\begin{lemma} \label{lem:cone-vector}
 Let $\Delta$ be an almost surely terminating, strongly connected pBPA.
 Then there is a positive vector $\vu \in \Rset_+^\Gamma$ such that
  $A \cdot \vu \le \vu$, where $\mathord{\le}$ is meant componentwise.
 All such vectors~$\vu$ satisfy $\frac{\umin}{\umax} \ge \pmin^{|\Gamma|}$,
  where $\pmin$ denotes the least rule probability in~$\Delta$,
  and $\umin$ and $\umax$ denote the least and the greatest component of~$\vu$, respectively.
\end{lemma}
\begin{proof}
  Let $X,Y \in \Gamma$.
  Since $\Delta$ is strongly connected, there is a sequence $X = X_1, X_2, \ldots, X_n = Y$ with $n \ge 1$ such that $X_i$ depends directly on~$X_{i+1}$
   for all $1 \le i \le n-1$.
  A straightforward induction on~$n$ shows that $A^n(X,Y) \ne 0$; i.e., $A$ is {\em irreducible}.
  The assumption that $\Delta$ is almost surely terminating implies that
   the spectral radius of~$A$ is less than or equal to one, see, e.g., Section 8.1 of~\cite{EY:RMC-SG-equations-JACM}.
  Perron-Frobenius theory (see, e.g., \cite{book:BermanP}) then implies that there is a positive vector $\vu \in \Rset_+^\Gamma$ such that
   $A \cdot \vu \le \vu$;
   e.g., one can take for~$\vu$ the dominant eigenvector of~$A$.

  Let $A \cdot \vu \le \vu$.
  It remains to show that $\frac{\umin}{\umax} \ge \pmin^{|\Gamma|}$.
  The proof is essentially given in~\cite{EKL10:SICOMP}, we repeat it for convenience.
  W.l.o.g.\ let $\Gamma = \{X_1, \ldots, X_{|\Gamma|}\}$.
  We write $\vu_i$ for $\vu(X_i)$.
  W.l.o.g.\ let $\vu_1 = \umax$ and $\vu_{|\Gamma|} = \umin$.
  Since $\Delta$ is strongly connected, there is a sequence $1 = r_1, r_2, \ldots, r_q = |\Gamma|$ with $q \le |\Gamma|$
   such that $X_{r_j}$ depends on~$X_{r_{j+1}}$ for all $j$.
  We have
   \[
    \frac{\umin}{\umax} = \frac{\vu_{|\Gamma|}}{\vu_1} = \frac{\vu_{r_q}}{\vu_{r_{q-1}}} \cdot \ldots \cdot \frac{\vu_{r_2}}{\vu_{r_1}} \,.
   \]
  By the pigeonhole principle there is $j$ with $2 \le j \le q$ such that
   \begin{equation}
    \frac{\umin}{\umax} \ge \left( \frac{\vu_s}{\vu_t} \right)^{q-1} \ge \left( \frac{\vu_s}{\vu_t} \right)^{|\Gamma|}\quad
      \text{where $s := r_j$ and $t := r_{j-1}$.} \label{eq:cone-vector-pigeonhole}
   \end{equation}
  We have $A \cdot \vu \le \vu$, which implies $A(X_s, X_t) \cdot \vu_t \le \vu_s$ and so $A(X_s, X_t) \le {\vu_s}/{\vu_t}$.
  On the other hand, since $X_s$ depends on~$X_t$, we clearly have $\pmin \le A(X_s, X_t)$.
  Combining those inequalities with~\eqref{eq:cone-vector-pigeonhole} yields $\frac{\umin}{\umax} \ge \left(A(X_s, X_t)\right)^{|\Gamma|} \ge \pmin^{|\Gamma|}$.
\qed
\end{proof}

Given a relaxed pBPA~$\Delta$ and vector $\vu \in \Rset_+^\Gamma$,
 we say that $\Delta$ is {\em $\vu$-progressive}, if $\Delta$ has, for all $X \in \Gamma$,
 a rule $X \btran{} \alpha$ such that $|\vu(X) - \numbsymb{\alpha} \thickdot \vu| \ge \umin/2$.
The following lemma states that, intuitively, any pBPA can be transformed into a $\vu$-progressive relaxed pBPA
 that is at least as fast but no more than ${|\Gamma|}$ times faster.
\begin{lemma} \label{lem:progressive}
 Let $\Delta$ be an almost surely terminating pBPA with stack alphabet~$\Gamma$.
 Let $\pmin$ denote the least rule probability in~$\Delta$, and let $\vu \in \Rset_+^\Gamma$ with $A_\Delta \cdot \vu \le \vu$.
 Then one can construct a $\vu$-progressive, almost surely terminating relaxed pBPA~$\Delta'$
  with stack alphabet~$\Gamma$ such that for all $\alpha_0\in\Gamma^*$ and for all $a \ge 0$
 \[
  \calP'(\mathbf{T}_{\alpha_0} \ge a) \quad \le \quad \calP( \mathbf{T}_{\alpha_0} \ge a ) \quad \le \quad \calP'(\mathbf{T}_{\alpha_0} \ge a / |\Gamma|)\,,
 \]
 where $\calP$ and $\calP'$ are the probability measures associated with $\Delta$ and $\Delta'$, respectively.
 Furthermore, the least rule probability in~$\Delta'$ is at least $\pmin^{|\Gamma|}$, and $A_{\Delta'} \cdot \vu \le \vu$.
 Finally, if $A_\Delta \cdot \vu = \vu$, then $A_{\Delta'} \cdot \vu = \vu$.
\end{lemma}
\begin{proof}
\newcommand{\Con}{\mathit{Con}}
A sequence of transitions $X_1 \btran{} \alpha_1, \ldots,  X_n \btran{} \alpha_n$ is called {\em derivation sequence from $X_1$ to $\alpha_n$},
 if for all $i \in \{2, \ldots, n\}$ the symbol $X_i \in \Gamma$ occurs in $\alpha_{i-1}$.
The {\em word induced} by a derivation sequence $X_1 \btran{} \alpha_1, \ldots,  X_n \btran{} \alpha_n$
 is obtained by taking $\alpha_1$, replacing an occurrence of~$X_2$ by~$\alpha_2$,
 then replacing an occurrence of~$X_3$ by~$\alpha_3$, etc., and finally replacing an occurrence of~$X_n$ by~$\alpha_n$.

Given a pBPA~$\Delta$ and a derivation sequence
 $s = \big( X_1 \btran{p_1} \alpha_1^1 X_2 \alpha_1^2, X_2 \btran{p_2} \alpha_2, \ldots,  X_n \btran{p_n} \alpha_{n} \big)$
 with $X_i \ne X_j$ for all $1 \le i < j \le n$,
we define the {\em contraction} $\Con(s)$ of~$s$, a set of $X_1$-transitions with possibly more than two symbols on the right hand side.
The contraction $\Con(s)$ will include a rule $X_1 \btran{} \gamma$, where $\gamma$ is the word induced by~$s$.
We define~$\Con(s)$ inductively over the length~$n$ of~$s$.
If $n = 1$, then $\Con(s) = \{X_1 \btran{p_1} \alpha_1^1 X_2 \alpha_1^2\}$.
If $n \ge 2$, let $s' = \big( X_2 \btran{p_2} \alpha_{2}, \ldots,  X_n \btran{p_n} \alpha_{n} \big)$ and define
 \begin{equation}
  \delta_2 := \left\{X_2 \btran{} \beta \mid \text{$X_2 \btran{} \beta$ is a rule in~$\Delta$} \right\}
    - \left\{ X_2 \btran{p2} \alpha_2 \right\}
    \cup \Con(s')\,; \label{eq:delta-2}
 \end{equation}
 i.e., $\delta_2$ is the set of $X_2$-transitions in~$\Delta$ with $X_2 \btran{p2} \alpha_2$ replaced by $\Con(s')$.
W.l.o.g.\ assume $\delta_2 = \{ X_2 \btran{q_{1}} \beta_1, \ldots, X_2 \btran{q_{k}} \beta_k \}$.
Then we define
 \[
  \Con(s) := \left\{ X_1 \btran{p_1q_1} \alpha_1^1 \beta_1 \alpha_1^2, \ldots, X_1 \btran{p_1q_k} \alpha_1^1 \beta_k \alpha_1^2 \right\}\,.
 \]
The following properties are easy to show by induction on~$n$:
\begin{itemize}
 \item[(a)]
  $\Con(s)$ contains $X_1 \btran{} \gamma$, where $\gamma$ is the word induced by~$s$.
 \item[(b)]
  The rule probabilities are at least $\pmin^n$.
 \item[(c)]
  Let $\Delta'$ be the relaxed pBPA obtained from~$\Delta$ by replacing $X_1 \btran{p_1} \alpha_1^1 X_2 \alpha_1^2$ with $\Con(s)$.
  Then each path in $M_{\Delta'}$ corresponds in a straightforward way to a path in~$M_\Delta$,
   namely to the path obtained by ``re-expanding'' the contractions.
  The corresponding path in~$M_\Delta$ has the same probability
   and is not shorter but at most $|\Gamma|$ times longer than the one in~$M_{\Delta'}$.
 \item[(d)]
  Let $\Delta'$ be as in~(c).
  Then $A_{\Delta'} \cdot \vu \le \vu$.
  Let us prove that explicitly.
  The induction hypothesis $n=1$ is trivial.
  For the induction step, using the definition for~$\delta_2$ in~\eqref{eq:delta-2} and
   $\delta_2 = \{ X_2 \btran{q_{1}} \beta_1, \ldots, X_2 \btran{q_{k}} \beta_k \}$,
   we know by the induction hypothesis that $\sum_{i=1}^k q_i \cdot \numbsymb{\beta_i} \thickdot \vu  \le \vu(X_2)$.
  This implies
   \begin{align*}
    \sum_{i=1}^k p_1 q_i \cdot \numbsymb{\alpha_1^1\beta_i\alpha_1^2} \thickdot \vu & \le p_1 \cdot \numbsymb{\alpha_1^1 X_2 \alpha_1^2} \thickdot \vu \,,
     \quad \text{and hence} \\
    \quad \left(A_{\Delta'} \cdot \vu\right)(X_1) & \le \left( A_{\Delta} \cdot \vu \right)(X_1) \le \vu(X_1)\,.
   \end{align*}
   Since $A_{\Delta}$ and $A_{\Delta'}$ may differ only in the $X_1$-row, we have $A_{\Delta'} \cdot \vu \le \vu$.
 \item[(e)]
  Let $\Delta'$ be as in (c) and~(d).
  If $A_{\Delta} \cdot \vu = \vu$, then $A_{\Delta'} \cdot \vu = \vu$.
  This follows as in~(d), with the inequality signs replaced by equality.
\end{itemize}

Associate to each symbol $X_1 \in \Gamma$ a shortest derivation sequence
 \[
  c(X_1) = \big( X_1 \btran{} \alpha_1, \ldots,  X_{n-1} \btran{} \alpha_{n-1}, X_n \btran{} \varepsilon \big)
 \]
  from $X_1$ to~$\varepsilon$.
Since $\Delta$ is almost surely terminating, the length of~$c(X_1)$ is at most~$|\Gamma|$ for all $X_1 \in \Gamma$.
Let $X_1 \in \Gamma$,
 and let $\gamma_1$ denote the word induced by $c(X_1)$,
 and let $\gamma_2$ denote the word induced by the derivation sequence $c_2(X_1) := \big( X_1 \btran{} \alpha_1, \ldots,  X_{n-1} \btran{} \alpha_{n-1} \big)$.
We have $\numbsymb{\gamma_2} \thickdot \vu =  \numbsymb{\gamma_1} \thickdot \vu + \vu(X_n) \ge \numbsymb{\gamma_1} \thickdot \vu + \umin$,
 so we can choose $\gamma \in \left\{\gamma_1, \gamma_2\right\}$ such that $|\vu(X_1) - \numbsymb{\gamma} \thickdot \vu| \ge \umin/2$.
Choose $\hat c(X_1) \in \{c(X_1), c_2(X_1)\}$ such that $\hat c(X_1)$ induces~$\gamma$.
(Of course, if $c_2(X_1)$ has length zero, take $\hat c(X_1) = c(X_1)$.)
Note that $(X_1 \btran{} \gamma) \in \Con(\hat c(X_1))$.

The relaxed pBPA~$\Delta'$ from the statement of the lemma is obtained by replacing, for all $X_1 \in \Gamma$,
 the first rule of~$\hat c(X_1)$ with $\Con(\hat c(X_1))$.
The properties (a)--(e) from above imply:
\begin{itemize}
 \item[(a)]
  The relaxed pBPA $\Delta'$ is $\vu$-progressive.
 \item[(b)]
  The rule probabilities are at least $\pmin^{|\Gamma|}$.
 \item[(c)]
  For each finite path $w'$ in~$M_{\Delta'}$ from some $\alpha_0 \in \Gamma^*$ to~$\varepsilon$ there is a
   finite path $w$ in~$M_{\Delta}$ from $\alpha_0$ to~$\varepsilon$ such that $|w'| \le |w| \le |\Gamma| \cdot |w'|$ and $\calP'(w') = \calP(w)$.
  Hence, $\calP'(\mathbf{T}_{\alpha_0} < a / |\Gamma|) \le \calP(\mathbf{T}_{\alpha_0} < a) \le \calP'(\mathbf{T}_{\alpha_0} < a)$ holds for all $a \ge 0$,
   which implies $\calP'(\mathbf{T}_{\alpha_0} \ge a) \le \calP( \mathbf{T}_{\alpha_0} \ge a ) \le \calP'(\mathbf{T}_{\alpha_0} \ge a / |\Gamma|)$.
 \item[(d)]
  We have $A_{\Delta'} \cdot \vu \le \vu$.
 \item[(e)]
  If $A_{\Delta} \cdot \vu = \vu$, then $A_{\Delta'} \cdot \vu = \vu$.
\end{itemize}
This completes the proof of the lemma.
\qed
\end{proof}

\begin{proposition} \label{prop:critical-u-progressive}
 Let $\Delta$ be an almost surely terminating relaxed pBPA with stack alphabet~$\Gamma$.
 Let $\vu \in \Rset_+^\Gamma$ be such that $\umax = 1$ and $A_\Delta \cdot \vu \le \vu$ and $\Delta$ is $\vu$-progressive.
 Let $\pmin$ denote the least rule probability in~$\Delta$.
 Let $C := 17 |\Gamma| / ( \pmin \cdot \umin^2 )$.
 Then for each $k \in \Nset_0$ there is $n_0 \in \Nset$ such that
   \begin{align*}
     \tailprob{\alpha_0}{ n^{2k+2} / (2 |\Gamma|) } \quad & \le\quad C / n
      && \qquad \text{for all $n \ge n_0$ and for all $\alpha_0 \in \Gamma^*$ with $1 \le |\alpha_0| \le n^k$.}
   \end{align*}
\end{proposition}
\begin{proof}
For each $X \in \Gamma$ we define a function $g_X: \Rset \to \Rset$ by setting
 \[
  g_X(\theta) := \sum_{X\btran{p} \alpha} p \cdot \exp(-\theta \cdot (-\vu(X) + \numbsymb{\alpha} \thickdot \vu))\,.
 \]
The following lemma states important properties of~$g_X$.
\begin{lemma} \label{lem:g-properties}
  The following holds for all $X \in \Gamma$:
  \begin{itemize}
    \item[(a)]
     For all $\theta > 0$ we have $1 = g_X(0) < g_X(\theta)$.
    \item[(b)]
     For all $\theta > 0$ we have $0 \le g_X'(0) < g_X'(\theta)$.
    \item[(c)]
     For all $\theta \ge 0$ we have $0 < g_X''(\theta)$. In particular, $g_X''(0) \ge \pmin \cdot \umin^2 / 4$.
  \end{itemize}
\end{lemma}
\begin{proof}[Proof of the lemma]\mbox{}
 \begin{itemize}
  \item[(a)]
   Clearly, $g_X(0) = 1$.
   The inequality $g_X(0) < g_X(\theta)$ follows from (b).
  \item[(b)]
   We have:
   \begin{align*}
    g_X(\theta)
    & = \sum_{X\btran{p} \alpha} p \cdot \exp(-\theta \cdot (-\vu(X) + \numbsymb{\alpha} \thickdot \vu)) \\
    g_X'(\theta)
    & = \sum_{X\btran{p} \alpha} p \cdot (\vu(X) - \numbsymb{\alpha} \thickdot \vu) \cdot \exp(-\theta \cdot (-\vu(X) + \numbsymb{\alpha} \thickdot \vu))
\intertext{Let $A(X)$ denote the $X$-row of~$A$, i.e., the vector $\vv\in \Rset^\Gamma$ such that $\vv(Y) = A(X,Y)$.
   Then $A \cdot \vu \le \vu$ implies
}
     g_X'(0)
     & = \sum_{X\btran{p} \alpha} p \cdot (\vu(X) - \numbsymb{\alpha} \thickdot \vu) \\
     & = \vu(X) - \sum_{X\btran{p} \alpha} p \cdot \numbsymb{\alpha} \thickdot \vu = \vu(X) - A(X) \thickdot \vu\\
     & \ge \vu(X) - \vu(X) = 0\,.
   \end{align*}
   The inequality $g_X'(0) < g_X'(\theta)$ follows from~(c).
  \item[(c)]
   We have
   \begin{align*}
    g_X''(\theta)
    & = \sum_{X\btran{p} \alpha} p \cdot (\vu(X) - \numbsymb{\alpha} \thickdot \vu)^2 \cdot \exp(-\theta \cdot (-\vu(X) + \numbsymb{\alpha} \thickdot \vu)) > 0\,.
   \end{align*}
   Since $\Delta$ is $\vu$-progressive, there is a rule $X \btran{p} \alpha$ with $|\vu(X) - \numbsymb{\alpha} \thickdot \vu| \ge \umin/2$.
   Hence, for $\theta = 0$ we have $g_X''(0) \ge \pmin \cdot \umin^2 / 4$.
 \end{itemize}
This proves the lemma.
\qed
\end{proof}

Let in the following $\theta > 0$.
Given a run $w\in \run(\alpha_0)$ and $i\geq 0$, we 
  write $\Xs{i}(w)$ for the symbol $X \in \Gamma$ for which $w(i) = X \alpha$.
Define
 \[
  \ms{i}_\theta(w) =
   \begin{cases} \displaystyle
           \exp(-\theta \cdot \numbsymb{w(i)} \thickdot \vu) \cdot \prod_{j=0}^{i-1} \frac{1}{g_{\Xs{j}(w)}(\theta)} & \text{if $i=0$ or $w(i-1) \ne \varepsilon$} \\
           \ms{i-1}_\theta(w)                                                                                  & \text{otherwise} \\
   \end{cases}
 \]

\begin{lemma} \label{lem:exponential-martingale}
$\ms{0}_\theta, \ms{1}_\theta, \ldots$ is a martingale.
\end{lemma}
\begin{proof}[Proof of the lemma]
Let us fix a path $v\in \fpath(\alpha_0)$ of length $i \ge 1$ and let $w$ be an arbitrary run of $\run(v)$.
First assume that $v(i-1)=X\alpha\in \Gamma\Gamma^*$.
Then we have:
\begin{align*}
 & \Ex{\ms{i}_{\theta} \;\middle\vert\; \run(v)} \\
 & = \Ex{\exp(-\theta \cdot \numbsymb{w(i)} \thickdot \vu) \cdot \prod_{j=0}^{i-1} \frac{1}{g_{\Xs{j}(w)}(\theta)} \;\middle\vert\; \run(v)} \\
 & = \sum_{X\btran{p}\alpha} p\cdot \exp\left(-\theta\cdot \left(\numbsymb{w(i-1)} - \vec{1}_X +  \numbsymb{\alpha}\right) \thickdot \vu\right) \cdot \prod_{j=0}^{i-1} \frac{1}{g_{\Xs{j}(w)}(\theta)} \\
 & = \sum_{X\btran{p}\alpha} p\cdot \exp\left(-\theta\cdot \left(\numbsymb{w(i-1)}\thickdot \vu - \vu(X) +  \numbsymb{\alpha}\thickdot \vu\right)\right) \cdot \prod_{j=0}^{i-1} \frac{1}{g_{\Xs{j}(w)}(\theta)} \\
 & = \exp\left(-\theta \cdot \numbsymb{w(i-1)}\thickdot \vu\right)\cdot \sum_{X\btran{p}\alpha} p\cdot \exp\left(-\theta \cdot \left(- \vu(X) +  \numbsymb{\alpha}\thickdot \vu\right)\right) \cdot \prod_{j=0}^{i-1} \frac{1}{g_{\Xs{j}(w)}(\theta)} \\
 & = \exp\left(-\theta \cdot \numbsymb{w(i-1)}\thickdot \vu\right)\cdot g_{\Xs{i-1}(w)}(\theta) \cdot \prod_{j=0}^{i-1} \frac{1}{g_{\Xs{j}(w)}(\theta)} \\
 & = \exp\left(-\theta \cdot \numbsymb{w(i-1)}\thickdot \vu\right)\cdot \prod_{j=0}^{i-2} \frac{1}{g_{\Xs{j}(w)}(\theta)} \\
 & = \ms{i-1}_\theta(w)\,.
\end{align*}
If $v(i-1)=\varepsilon$, then for every $w\in \run(v)$ we have $\ms{i}_\theta(w)=\ms{i-1}_\theta(w)$.
Hence, $\ms{0}_\theta, \ms{1}_\theta, \ldots$ is a martingale.
\qed
\end{proof}

Since $\theta > 0$ and since $g_{\Xs{j}(w)}(\theta) \ge 1$ by Lemma~\ref{lem:g-properties}(a),
 we have $0 \le \ms{i}_\theta(w) \le 1$, so the martingale is bounded.
Since, furthermore, $\mathbf{T}_{\alpha_0}$ (we write only $\termt{}$ in the following) is finite with probability~$1$,
 it follows using Doob's Optional-Stopping Theorem (see Theorem~10.10~(ii) of~\cite{book:Williams})
 that $\ms{0}_\theta = \Ex{\ms{\termt{}}_\theta}$.
Hence we have for each $n \in \Nset$:
\begin{align*}
 & \quad   \exp(-\theta \cdot \umax \cdot n^k) \\
 & \le \exp(-\theta \cdot \vu \thickdot \numbsymb{\alpha_0}) = \ms{0}_\theta \\ 
 & = \Ex{\ms{\termt{}}_\theta} && \text{(by optional-stopping)} \\
 & = \Ex{\exp(-\theta \cdot 0) \cdot \prod_{j=0}^{\termt{}-1} \frac{1}{g_{\Xs{j}}(\theta)}} \\
 & = \Ex{\prod_{j=0}^{\termt{}-1} \frac{1}{g_{\Xs{j}}(\theta)}} \\
 & \le \Ex{\frac{1}{g_X(\theta)^\termt{}}} && \text{(for some $X \in \Gamma$)} \\
 & = \sum_{i=0}^\infty \calP(\termt{} = i) \cdot \frac{1}{g_X(\theta)^i} \\
 & \le \sum_{i=0}^{\left\lceil n^{2k+2} / (2 |\Gamma|) \right\rceil-1} \calP(\termt{} = i) \cdot 1  && \text{(Lemma~\ref{lem:g-properties}~(a))} \\
 & \quad + \sum_{i=\left\lceil n^{2k+2} / (2 |\Gamma|) \right\rceil}^\infty \calP(\termt{} = i) \cdot \frac{1}{g_X(\theta)^{ n^{2k+2} / (2 |\Gamma|) }} \\
 & = 1 - \calP(\termt{} \ge n^{2k+2}/(2 |\Gamma|)) \\
 & \quad \mbox{} + \calP(\termt{} \ge n^{2k+2}/(2 |\Gamma|)) \cdot \frac{1}{g_X(\theta)^{n^{2k+2}/(2 |\Gamma|)}}
\end{align*}
Rearranging the inequality, we obtain
\begin{equation}
 \calP(\termt{} \ge n^{2k+2}/(2 |\Gamma|)) \le \frac{1 - \exp(-\theta \cdot \umax \cdot n^{k})}{1 - g_X(\theta)^{-n^{2k+2}/(2 |\Gamma|)}} \;.
  \label{eq:before-hopital}
\end{equation}
For the following we set $\theta = n^{-(k+1)}$.
We want to give an upper bound for the right hand side of~\eqref{eq:before-hopital}.
To this end we will show:
 \begin{equation}
  \lim_{n\to\infty} \frac{\left(1 - \exp(-n^{-(k+1)} \cdot \umax \cdot n^{k})\right) \cdot n}{1 - g_X(n^{-(k+1)})^{-n^{2(k+1)}/ (2 |\Gamma|)}}
  \le \frac{1}{1 - \exp \left( - \pmin \cdot \umin^2 / (16 |\Gamma|) \right) }
  \,. \label{eq:quotient}
 \end{equation}
Combining \eqref{eq:before-hopital} with \eqref{eq:quotient}, we obtain
 \begin{align*}
  \limsup_{n \to \infty} \ n \cdot \calP(\termt{} \ge n^{2k+2}/(2 |\Gamma|))
   & \le \frac{1}{1 - \exp \left( - \pmin \cdot \umin^2 / (16 |\Gamma|) \right) } \\
   & < \frac{1}{1 - \left(1 - \frac{16}{17} \cdot \left( \pmin \cdot \umin^2 / (16 |\Gamma|) \right) \right) } \\
   & = 17 |\Gamma| / ( \pmin \cdot \umin^2 )
   \,,
 \end{align*}
which implies the proposition.

To prove~\eqref{eq:quotient}, we compute limits for the nominator and the denominator separately.
For the nominator, we use l'Hopital's rule to obtain:
 \begin{align*}
  \lim_{n\to\infty} \frac{1 - \exp(-\umax \cdot n^{-1})}{n^{-1}}
   & = \lim_{n\to\infty} \frac{-\umax \cdot n^{-2} \cdot \exp(-\umax \cdot n^{-1})}{-n^{-2}} = \umax = 1\,.
 \end{align*}
For the denominator of~\eqref{eq:quotient} we consider first the following limit:
 \begin{align*}
  & \lim_{n\to\infty} \frac1{2 |\Gamma|} \cdot n^{2(k+1)} \cdot \ln g_X(n^{-(k+1)}) \\
  & = \frac1{2 |\Gamma|} \lim_{n\to\infty} \frac{\ln g_X(n^{-(k+1)})}{n^{-2(k+1)}} \\
  & = \frac1{2 |\Gamma|} \lim_{n\to\infty} \frac{g_X'(n^{-(k+1)}) \cdot \left(-(k+1)\right) \cdot n^{-k-2}}{g_X(n^{-(k+1)}) \cdot \left( -2(k+1) \right) \cdot n^{-2k-3}}
       && \text{(l'Hopital's rule)} \\
  & = \frac1{4 |\Gamma|} \lim_{n\to\infty} \frac{g_X'(n^{-(k+1)})}{n^{-(k+1)}}
       && \text{(by Lemma~\ref{lem:g-properties}~(a))\,.} \\
\intertext{If $g_X'(0) > 0$, then the limit is $+\infty$. Otherwise, by Lemma~\ref{lem:g-properties}~(b), we have $g_X'(0) = 0$ and hence}
  & = \frac1{4 |\Gamma|} \lim_{n\to\infty} \frac{g_X''(n^{-(k+1)}) \cdot \left(-(k+1)\right) \cdot n^{-k-2}}{\left(-(k+1)\right) \cdot n^{-k-2}}
       && \text{(l'Hopital's rule)} \\
  & = \frac1{4 |\Gamma|} g_X''(0) \ge \pmin \cdot \umin^2 / (16 |\Gamma|)  && \text{(by Lemma~\ref{lem:g-properties}~(c))\,.}
 \end{align*}
This proves~\eqref{eq:quotient}
 and thus completes the proof of Proposition~\ref{prop:critical-u-progressive}.
\qed
\end{proof}

The following lemma serves as induction base for the proof of Proposition~\ref{prop:term-upper-critical}.
\begin{lemma} \label{lem:critical-strongly-connected}
  Let $\Delta$ be an almost surely terminating pBPA with stack alphabet~$\Gamma$.
  Assume that all SCCs of~$\Delta$ are bottom SCCs.
  Let $\pmin$ denote the least rule probability in~$\Delta$.
  Let $D := 17 |\Gamma| / \pmin^{3|\Gamma|}$.
  Then for each $k \in \Nset_0$ there is $n_0 \in \Nset$ such that
   \begin{align*}
     \tailprob{\alpha_0}{ n^{2k+2} / 2 } \quad & \le\quad D / n
      && \quad \text{for all $n \ge n_0$ and for all $\alpha_0 \in \Gamma^*$ with $1 \le |\alpha_0| \le n^k$.}
   \end{align*}
\end{lemma}
\begin{proof}
Decompose $\Gamma$ into its SCCs, say $\Gamma = \Gamma_1 \cup \cdots \cup \Gamma_s$,
 and let the pBPA~$\Delta_i$ be obtained by restricting~$\Delta$ to the $\Gamma_i$-symbols.
For each $i \in \{1, \ldots, s\}$, Lemma~\ref{lem:cone-vector} gives a vector $\vu_i \in \Rset_+^{\Gamma_i}$.
W.l.o.g.\ we can assume for each~$i$ that the largest component of~$\vu_i$ is equal to~$1$,
 because $\vu_i$ can be multiplied with any positive scalar without changing the properties guaranteed by Lemma~\ref{lem:cone-vector}.
If the vectors~$\vu_i$ are assembled (in the obvious way) to the vector $\vu \in \Rset_+^\Gamma$,
 the assertions of Lemma~\ref{lem:cone-vector} carry over;
 i.e., we have $A_\Delta \cdot \vu \le \vu$ and $\umax = 1$ and $\umin \ge \pmin^{|\Gamma|}$.
Let $\Delta'$ be the $\vu$-progressive relaxed pBPA from Lemma~\ref{lem:progressive},
 and denote by~$\calP'$ and $\pmin'$ its associated probability measure and least rule probability, respectively.
Then we have:
\begin{align*}
 \tailprob{\alpha_0}{ n^{2k+2} / 2 }
 & \le \calP'( \termt{\alpha_0} \ge n^{2k+2} / (2 |\Gamma|) )   && \text{(by Lemma~\ref{lem:progressive})} \\
 & \le 17 |\Gamma| / ( \pmin' \cdot \umin^2 \cdot n )           && \text{(by Proposition~\ref{prop:critical-u-progressive})} \\
 & \le 17 |\Gamma| / ( \pmin' \cdot \pmin^{2|\Gamma|} \cdot n ) && \text{(as argued above)} \\
 & \le 17 |\Gamma| / ( \pmin^{3|\Gamma|} \cdot n )              && \text{(by Lemma~\ref{lem:progressive})\,.}
\end{align*}
\qed
\end{proof}

\newcommand{\GammaH}{\Gamma_{\it high}}
\newcommand{\GammaL}{\Gamma_{\it low}}
\newcommand{\DeltaH}{\Delta_{\it high}}
\newcommand{\DeltaL}{\Delta_{\it low}}
Now we are ready to prove Proposition~\ref{prop:term-upper-critical}, which is restated here.
\begin{qproposition}{\ref{prop:term-upper-critical}}
 \stmtproptermuppercritical
\end{qproposition}
\begin{proof}
 Let $D$ be the $D$ from Lemma~\ref{lem:critical-strongly-connected}.
 We show by induction on~$h$:
 \begin{equation}
   \tailprob{X_0}{n^{2^{h+1}-2}} \le \frac{h D}{n} \quad \text{for almost all $n\in\Nset$.} \label{eq:ct-induction}
 \end{equation}
 Note that \eqref{eq:ct-induction} implies the proposition.
 The case $h=1$ (induction base) is implied by Lemma~\ref{lem:critical-strongly-connected}.
 Let $h \ge 2$.
 Partition $\Gamma$ into $\GammaH \cup \GammaL$ such that $\GammaL$
  contains the variables of the SCCs of depth~$h$ in the DAG of SCCs,
  and $\GammaH$ contains the other variables (in ``higher'' SCCs).
 If $X_0 \in \GammaL$, then we can restrict~$\Delta$ to the variables that are in the same SCC as~$X_0$,
  and Lemma~\ref{lem:critical-strongly-connected} implies~\eqref{eq:ct-induction}.
 So we can assume $X_0 \in \GammaH$.

 Assume for a moment that $\tailprob{X_0}{n^{2^{h+1}-2}}$ holds for a run~$w \in \run(X_0)$;
  i.e., we have:
 \begin{align*}
  n^{2^{h+1}-2} \quad & \le \quad |\{ i \in \Nset_0 \mid w(i) \in \Gamma\Gamma^*\}| \\
   & = \quad |\{ i \in \Nset_0 \mid w(i) \in \GammaH\Gamma^*\}| + |\{ i \in \Nset_0 \mid w(i) \in \GammaL\Gamma^*\}|\,.
 \end{align*}
 It follows that one of the following events is true for~$w$:
 \begin{itemize}
  \item[(a)]
   At least $n^{2^h-2}$ steps in~$w$ have a $\GammaH$-symbol on top of the stack.
   More formally,
    \[
     |\{ i \in \Nset_0 \mid w(i) \in \GammaH\Gamma^*\}| \ge n^{2^h-2}\,.
    \]
  \item[(b)]
   Event~(a) is not true, but at least $n^{2^{h+1}-2} - n^{2^{h}-2}$ steps in~$w$ have a $\GammaL$-symbol on top of the stack.
   More formally,
    \begin{align*}
     |\{ i \in \Nset_0 \mid w(i) \in \GammaH\Gamma^*\}| & < n^{2^h-2} \quad \text{and} \\
     |\{ i \in \Nset_0 \mid w(i) \in \GammaL\Gamma^*\}| & \ge n^{2^{h+1}-2} - n^{2^{h}-2}\,.
    \end{align*}
 \end{itemize}
 In order to give bounds on the probabilities of events (a) and~(b),
  it is convenient to ``reshuffle'' the execution order of runs in the following way:
 Whenever a rule $X \btran{} \alpha$ is executed, we do not replace the $X$-symbol on top of the stack by~$\alpha$,
  but instead we push only the $\GammaH$-symbols in~$\alpha$ on top of the stack, whereas the $\GammaL$-symbols in~$\alpha$
  are added to the {\em bottom} of the stack.
 Since $\Delta$ is a pBPA and thus does not have control states,
  the reshuffling of the execution order does not influence the distribution of the termination time.
 The advantage of this execution order is that each run can be decomposed into two phases:
 \begin{itemize}
  \item[(1)]
   In the first phase, the symbol on the top of the stack is always a $\GammaH$-symbol.
   When rules are executed, $\GammaL$-symbols may be produced, which are added to the bottom of the stack.
  \item[(2)]
   In the second phase, the stack consists of $\GammaL$-symbols exclusively.
   Notice that by definition of $\GammaL$, no new $\GammaH$-symbols can be produced.
 \end{itemize}
 In terms of those phases, the above events (a) and~(b) can be reformulated as follows:
 \begin{itemize}
  \item[(a)]
   The first phase of~$w$ consists of at least $n^{2^h-2}$ steps.
   The probability of this event is equal to
    \[
     \calP_{\DeltaH}(\termt{X_0} \ge n^{2^h-2})\,,
    \]
   where $\DeltaH$ is the pBPA obtained from~$\Delta$ by deleting all $\GammaL$-symbols
   from the right hand sides of the rules and deleting all rules with $\GammaL$-symbols on the left hand side,
   and $\calP_{\DeltaH}$ is its associated probability measure.
  \item[(b)]
   The first phase of~$w$ consists of fewer than $n^{2^h-2}$ steps
    (which implies that at most $n^{2^h-2}$ $\GammaL$-symbols are produced during the first phase),
    and the second phase consists of at least $n^{2^{h+1}-2} - n^{2^{h}-2}$ steps.
   Therefore, the probability of the event~(b) is at most
    \[
     \max \left\{ \calP_{\DeltaL}(\termt{\alpha_0} \ge n^{2^{h+1}-2} - n^{2^h-2})
      \;\middle\vert\; \alpha_0 \in \GammaL^*,\; 1 \le |\alpha_0| \le n^{2^h-2} \right\}\,,
    \]
    where $\DeltaL$ is the pBPA~$\Delta$ restricted to the $\GammaL$-symbols,
     and $\calP_{\DeltaL}$ is its associated probability measure.
   Notice that $n^{2^{h+1}-2} - n^{2^h-2} \ge n^{2^{h+1}-2} / 2$ for large enough~$n$.
   Furthermore, by the definition of~$\GammaL$, the SCCs of~$\DeltaL$ are all bottom SCCs.
   Hence, by Lemma~\ref{lem:critical-strongly-connected}, the above maximum is at most~$D/n$.
 \end{itemize}
 Summing up, we have for almost all $n \in \Nset$:
 \begin{align*}
  \tailprob{X_0}{n^{2^{h+1}-2}}
  & \le \calP(\text{event~(a)}) + \calP(\text{event~(b)}) \\
  & \le \calP_{\DeltaH}(\termt{X_0} \ge n^{2^h-2})   +   D/n     && \text{(as argued above)} \\
  & \le \frac{(h-1)D}{n} + \frac{D}{n} = \frac{hD}{n}           && \text{(by the induction hypothesis).}
 \end{align*}
 This completes the induction proof.
\qed
\end{proof}

\subsection{Proof of Proposition~\ref{prop:term-lower-critical}} \label{app:term-lower-critical}

The proof of Proposition~\ref{prop:term-lower-critical} is similar to the proof of Proposition~\ref{prop:term-upper-critical}
 from the previous subsection.
Here is a restatement of Proposition~\ref{prop:term-lower-critical}.
\begin{qproposition}{\ref{prop:term-lower-critical}}
 \stmtproptermlowercritical
\end{qproposition}

\begin{proof}
 For a square matrix~$M$ denote by $\rho(M)$ the spectral radius of~$M$, i.e., the greatest absolute value of its eigenvectors.
 Let $A_\Delta$ be the matrix from the previous subsection.
 We claim:
 \begin{equation}
   \rho(A_\Delta) = 1\,. \label{eq:spectral-radius}
 \end{equation}
 The assumption that $\Delta$ is almost surely terminating implies that $\rho(A_\Delta) \le 1$,
  see, e.g., Section 8.1 of~\cite{EY:RMC-SG-equations-JACM}.
 Assume for a contradiction that $\rho(A_\Delta) < 1$.
 Using standard theory of nonnegative matrices (see, e.g., \cite{book:BermanP}),
  this implies that the matrix inverse $B := (I - A_\Delta)^{-1}$ (here, $I$ denotes the identity matrix) exists; i.e., $B$ is finite in all components.
 It is shown in~\cite{EKM:prob-PDA-expectations} that $E[X_0] = (B \cdot \vone)(X_0)$ (here, $\vone$ denotes the vector with $\vone(X) = 1$ for all $X$).
 This is a contradiction to our assumption that $E[X_0] = \infty$.
 Hence, \eqref{eq:spectral-radius} is proved.

 It follows from~\eqref{eq:spectral-radius} and standard theory of nonnegative matrices~\cite{book:BermanP}
  that $A_\Delta$ has a principal submatrix, say $A'$, which is irreducible and satisfies $\rho(A') = 1$.
 Let $\Gamma'$ be the subset of~$\Gamma$ such that $A'$ is obtained from~$A$ by deleting all rows and columns
  which are not indexed by~$\Gamma'$.
 Let $\Delta'$ be the pBPA with stack alphabet~$\Gamma'$ such that $\Delta'$ is obtained from~$\Delta$ by
  removing all rules with symbols from $\Gamma \setminus \Gamma'$ on the left hand side and
  removing all symbols from $\Gamma \setminus \Gamma'$ from all right hand sides.
 Clearly, $A_{\Delta'} = A'$, so $\rho(A_{\Delta'}) = 1$ and $A_{\Delta'}$ is irreducible.
 Since $\Delta'$ is a sub-pBPA of~$\Delta$ and $X_0$ depends on all symbols in~$\Gamma'$,
  it suffices to prove the proposition for~$\Delta'$ and an arbitrary start symbol $X_0' \in \Gamma'$.

 Therefore, w.l.o.g.\ we can assume in the following that $A_\Delta = A$ is irreducible.
 Then it follows, using~\eqref{eq:spectral-radius} and Perron-Frobenius theory~\cite{book:BermanP},
  that there is a positive vector $\vu \in \Rset_+^\Gamma$ such that $A \cdot \vu = \vu$.
 W.l.o.g.\ we assume $\vu(X_0) = 1$.
 Using Lemma~\ref{lem:progressive} we can assume w.l.o.g.\ that $\Delta$ is $\vu$-progressive.
 (The pBPA~$\Delta$ may be relaxed.)

 As in the proof of Proposition~\ref{prop:critical-u-progressive}, for each $X \in \Gamma$ we define a function $g_X: \Rset \to \Rset$ by setting
 \[
  g_X(\theta) := \sum_{X\btran{p} \alpha} p \cdot \exp(-\theta \cdot (-\vu(X) + \numbsymb{\alpha} \thickdot \vu))\,.
 \]
 The following lemma states some properties of~$g_X$.
 \begin{lemma} \label{lem:g-properties-lower}
   The following holds for all $X \in \Gamma$:
   \begin{itemize}
     \item[(a)]
      For all $\theta > 0$ we have $1 = g_X(0) < g_X(\theta)$.
     \item[(b)]
      For all $\theta > 0$ we have $0 = g_X'(0) < g_X'(\theta)$.
     \item[(c)]
      For all $\theta \ge 0$ we have $0 < g_X''(\theta)$.
     \item[(d)]
      There is $c_2 > 0$ such that for all $0 < \theta \le 1$ we have $g_X'(\theta) \le c_2 \theta$.
     \item[(e)]
      There is $c_3 > 1$ such that for all $n \in \Nset$ we have $g_X(1/\sqrt{n})^n \ge c_3$.
     \item[(f)]
      There is $c_4 > 0$ such that for all $n \in \Nset$ we have $\frac{1/n}{1 - 1/g_X(1/\sqrt{n})} \le c_4$.
   \end{itemize}
 \end{lemma}
 \begin{proof}[of the lemma]
  The proof of items (a)--(c) follows exactly the proof of Lemma~\ref{lem:g-properties} and is therefore omitted.
  (For the equality $0 = g_X'(0)$ in~(b) one uses $A \cdot \vu = \vu$.)
  \begin{itemize}
   \item[(d)]
    It suffices to prove that $g_X'(\theta) / \theta$ is bounded for $\theta \to 0$.
    Using l'Hopital's rule we have $\lim_{\theta \to 0} g_X'(\theta) / \theta = g_X''(0) > 0$.
   \item[(e)]
    Clearly, we have $g_X(1/\sqrt{n})^n > 1$ for all~$n$.
    Furthermore, we have:
    \begin{align*}
     \lim_{n \to \infty} \ln g_X(1/\sqrt{n})^n
     & = \lim_{n \to \infty} \frac{\ln g_X(n^{-1/2})}{1/n} \\
     & = \frac12 \lim_{n \to \infty} \frac{g_X'(n^{-1/2})}{n^{-1/2}} && \text{(l'Hopital's rule)} \\
     & = \frac{g_X''(0)}{2} && \text{(l'Hopital's rule)} \\
     & > 0 && \text{(by~(c))}
    \end{align*}
    Hence the claim follows.
   \item[(f)]
    The claim follows again from l'Hopital's rule:
    \begin{align*}
     \lim_{n \to \infty} \frac{1/n}{1 - 1/g_X(n^{-1/2})}
     & = \lim_{n \to \infty} \frac{-1/n^2}{(1/g_X(n^{-1/2}))^2 \cdot g_X'(n^{-1/2}) \cdot (-1/2) n^{-3/2}} \\
     & = \lim_{n \to \infty} \frac{2 n^{-1/2}}{g_X'(n^{-1/2})} = \frac{2}{g_X''(0)} < \infty
    \end{align*}
  \end{itemize}
 This completes the proof of the lemma.
 \qed
 \end{proof}

 Let in the following $\theta > 0$.
 As in the proof of Proposition~\ref{prop:critical-u-progressive},
  given a run $w\in \run(X_0)$ and $i\geq 0$, we write $\Xs{i}(w)$ for the symbol $X \in \Gamma$ for which $w(i) = X \alpha$.
 Define
  \[
   \ms{i}_\theta(w) =
    \begin{cases} \displaystyle
            \exp(-\theta \cdot \numbsymb{w(i)} \thickdot \vu) \cdot \prod_{j=0}^{i-1} \frac{1}{g_{\Xs{j}(w)}(\theta)}
                                                        & \text{if $i=0$ or $w(i-1) \ne \varepsilon$} \\
            \ms{i-1}_\theta(w)                          & \text{otherwise} \\
    \end{cases}
  \]
 As in Lemma~\ref{lem:exponential-martingale}, one can show that the sequence $\ms{0}_\theta, \ms{1}_\theta, \ldots$ is a martingale.
 As in the proof of Proposition~\ref{prop:critical-u-progressive},
  Doob's Optional-Stopping Theorem implies $\exp(-\theta) = \ms{0}_\theta = \Ex{\ms{\termt{X_0}}_\theta}$.
 Hence we have for each $n \in \Nset$ (writing $\termt{}$ for~$\termt{X_0}$):
 \begin{align}
        \exp(-\theta)
  & = \Ex{\ms{\termt{}}_\theta} && \text{(by optional-stopping)} \nonumber\\
  & = \Ex{\exp(-\theta \cdot 0) \cdot \prod_{j=0}^{\termt{}-1} \frac{1}{g_{\Xs{j}}(\theta)}} \nonumber\\
  & = \Ex{\prod_{j=0}^{\termt{}-1} \frac{1}{g_{\Xs{j}}(\theta)}} \nonumber
\intertext{Taking, on both sides, the derivative with respect to~$\theta$ yields}
        \exp(-\theta)
  & \le \sum_{i=1}^\infty i \cdot \calP(\termt{} = i) \cdot \frac{g_{1,\theta}'(\theta)}{g_{0,\theta}(\theta)^{i+1}}\,, \label{eq:lower-1}
 \end{align}
 where $g_{0,\theta} = g_X$ and $g_{1,\theta} = g_Y$ for some $X,Y \in \Gamma$ possibly depending on~$\theta$.
 The following lemma bounds an ``upper'' subseries of the right-hand-side of~\eqref{eq:lower-1}.
 \begin{lemma} \label{lem:upper-subseries}
  For all $\varepsilon > 0$ there is $a \in \Nset$ such that for all $n \in \Nset$ and $\theta = 1/\sqrt{n}$ we have
  \[
   \sum_{i=a n + 1}^\infty i \cdot \calP(\termt{} = i) \cdot \frac{g_{1,\theta}'(\theta)}{g_{0,\theta}(\theta)^{i+1}} \quad \le \quad \varepsilon\,.
  \]
 \end{lemma}
 \begin{proof}[of the lemma]
 By rearranging the series we get for all $n \in \Nset$ and $\theta = 1/\sqrt{n}$:
 \begin{align*}
   & \sum_{i=a n + 1}^\infty i \cdot \calP(\termt{} = i) \cdot \frac{g_{1,\theta}'(\theta)}{g_{0,\theta}(\theta)^{i+1}} \\
   & \le \sum_{i=0}^{a n - 1} \frac{\calP(\termt{} > a n) \cdot g_{1,\theta}'(\theta)}{g_{0,\theta}(\theta)^{a n + 2}}
       + \sum_{i=a n}^\infty \frac{\calP(\termt{} > i) \cdot g_{1,\theta}'(\theta)}{g_{0,\theta}(\theta)^{i + 2}} \\
   & \le \underbrace{\frac{a n \cdot \calP(\termt{} > a n) \cdot g_{1,\theta}'(\theta)}{g_{0,\theta}(\theta)^{a n}}}_{=:q_1}
       + \underbrace{\sum_{i=a n}^\infty \frac{\calP(\termt{} > i) \cdot g_{1,\theta}'(\theta)}{g_{0,\theta}(\theta)^{i}}}_{=:q_2} \\
 \end{align*}
 We bound $q_1$ and~$q_2$ separately.
 By Proposition~\ref{prop:term-upper-critical} there is $c_1 > 0$ such that $\calP(\termt{} > k) \le c_1/\sqrt{k}$.
 Hence we have, using Lemma~\ref{lem:g-properties-lower}~(d),~(e):
 \begin{align*}
  q_1 & \le \frac{\sqrt{a n} \cdot c_1 \cdot c_2 / \sqrt{n}}{c_3^a} \le \frac{c_1 c_2 \sqrt{a}}{c_3^a}\,, && \text{and similarly,} \\
  q_2 & \le \frac{c_1}{\sqrt{a n}} \cdot \frac{c_2}{\sqrt{n}} \cdot \sum_{i=a n}^\infty \frac{1}{g_{0,\theta}(\theta)^i} \\
      & = \frac{c_1 c_2}{\sqrt{a} \cdot n \cdot g_{0,\theta}(\theta)^{a n} \cdot \left( 1 - 1/g_{0,\theta}(\theta) \right)} \\
      & \le \frac{c_1 c_2 c_4}{\sqrt{a} \cdot c_3^a} && \text{(by Lemma~\ref{lem:g-properties-lower}~(e),~(f))\,.}
 \end{align*}
 These bounds on $q_1$ and $q_2$ can be made arbitrarily small by choosing~$a$ large enough.
 This completes the proof of the lemma.
 \qed
 \end{proof}
 This lemma implies a first lower bound on the distribution of~$\termt{}$:
 \begin{lemma} \label{lem:lower-first-bound}
  For any $c>0$ there is $s \in \Nset$ such that for all $n \in \Nset$ we have:
  \[
   \sum_{i=1}^{s n} i \cdot \calP(\termt{} = i) \ge c \sqrt{n}\,.
  \]
 \end{lemma}
 \begin{proof}[of the lemma]
  Let $a \in \Nset$ be the number from Lemma~\ref{lem:upper-subseries} for $\varepsilon = \exp(-1)/2$.
  For all $n \in \Nset$ and $\theta = 1/\sqrt{n}$ we have:
  \begin{align*}
    & g_{1,\theta}'(\theta) \cdot \sum_{i=1}^{a n}  i \cdot \calP(\termt{} = i) \\
    & \ge \sum_{i=1}^{a n}  i \cdot \calP(\termt{} = i) \cdot \frac{g_{1,\theta}'(\theta)}{g_{0,\theta}(\theta)^{i+1}} \\
    & \ge \exp(-\theta) - \varepsilon && \text{(by~\eqref{eq:lower-1} and Lemma~\ref{lem:upper-subseries})} \\
    & \ge \exp(-1) - \varepsilon = \varepsilon && \text{(by the choice of~$\varepsilon$),}
  \end{align*}
  so, with Lemma~\ref{lem:g-properties-lower}~(d) we have for all $n \in \Nset$:
  \[
   \sum_{i=1}^{a n} i \cdot \calP(\termt{} = i) \ge \frac{\varepsilon}{c_2} \sqrt{n}\,.
  \]
  For the given number~$c>0$, choose $s := a \lceil c c_2 / \varepsilon \rceil^2$.
  Then it follows for all $m \in \Nset$:
  \[
   \sum_{i=1}^{s m}  i \cdot \calP(\termt{} = i) \ge c \sqrt{m}\,,
  \]
  which proves the lemma.
 \qed
 \end{proof}

 Now we can complete the proof of the proposition.
 By Proposition~\ref{prop:term-upper-critical} there is $c_1 > 0$ such that $\calP(\termt{} > n) \le c_1/\sqrt{n}$ for all $n \in \Nset$.
 By Lemma~\ref{lem:lower-first-bound}, there is $s \in \Nset$ such that
 \[
  \sum_{i=1}^{s n} i \cdot \calP(\termt{} = i) \ge (2 c_1 + 2) \sqrt{n} \quad \text{for all $n \in \Nset$.}
 \]
 We have for all $n \in \Nset$:
 \begin{align*}
  \sum_{i=n}^{s n} i \cdot \calP(\termt{} = i)
  & \ge \sum_{i=1}^{s n} i \cdot \calP(\termt{} = i) - \sum_{i=1}^{n} i \cdot \calP(\termt{} = i) \\
  & \ge (2 c_1 + 2) \sqrt{n} - \sum_{i=0}^n \calP(\termt{} > i) && \text{(by the choice of~$s$ above)} \\
  & \ge (2 c_1 + 2) \sqrt{n} - 1 - \sum_{i=1}^n \frac{c_1}{\sqrt{i}} && \text{(by the choice of~$c_1$ above)} \\
  & \ge (2 c_1 + 1) \sqrt{n} - \int_{0}^n \frac{c_1}{\sqrt{i}} \,d i \\
  & = (2 c_1 + 1) \sqrt{n} - 2 c_1 \sqrt{n} \\
  & = \sqrt{n}
\intertext{It follows:}
  s n \calP(\termt{} \ge n)
  & \ge s n \sum_{i=n}^{s n} \calP(\termt{} = i)
    \ge \sum_{i=n}^{s n} i \cdot \calP(\termt{} = i) \\
  & \ge \sqrt{n} && \text{(by the computation above)}
\intertext{Hence we have}
  \calP(\termt{} \ge n)
  & \ge \frac{1/s}{\sqrt{n}}\,,
 \end{align*}
 which completes the proof of the proposition.
\qed
\end{proof}

\subsection{Proof of Proposition~\ref{prop:heavy-tail}}
\label{app-optimal}

Here is a restatement of Proposition~\ref{prop:heavy-tail}.
\begin{qproposition}{\ref{prop:heavy-tail}}
 \stmtpropheavytail
\end{qproposition}
\begin{proof}
Observe that the third statement implies the second statement, since
\[
 E[X_h] = \sum_{n=1}^\infty \tailprob{X_h}{n}
        \ge \sum_{n=1}^\infty c_h \cdot n^{-1/2^h} \ge \sum_{n=1}^\infty c_h / n = \infty\;.
\]
We proceed by induction on~$h$.
Let $h=1$.
The pBPA $\Delta_1$ is equivalent to a random walk on $\{0, 1, 2, \ldots\}$,
 started at~$1$, with an absorbing barrier at~$0$.
It is well-known (see, e.g., \cite{Chung:book}) that the probability that the random walk finally reaches~$0$ is~$1$,
 but that there is $c_1 > 0$ such that the probability that the random has not reached~$0$
 after $n$ steps is at least $c_1 / \sqrt{n}$.
Hence $[X_1] = 1$ and $\tailprob{X_1}{n} \ge c_1 / \sqrt{n} = c_1 \cdot n^{-1/2}$.

Let $h > 1$.
The behavior of~$\Delta_h$ can be described in terms of a random walk~$W_h$
 whose states correspond to the number of $X_h$-symbols in the stack.
Whenever an $X_h$-symbol is on top of the stack, the total number of $X_h$-symbols in the stack
 increases by~$1$ with probability~$1/2$, or decreases by~$1$ with probability~$1/2$,
 very much like the random walk equivalent to~$\Delta_1$.
In the second case (i.e., the rule $X_h \btran{1/2} X_{h-1}$ is taken),
 the random walk~$W_h$ resumes only after a run of~$\Delta_{h-1}$
 (started with a single $X_{h-1}$-symbol) has terminated.
By the induction hypothesis, $[X_{h-1}] = 1$,
 so with probability~$1$ all spawned ``sub-runs'' of~$\Delta_{h-1}$ terminate.
Since $W_h$ also terminates with probability~$1$, it follows $[X_h] = 1$.

It remains to show that there is $c_h > 0$ with
 $\tailprob{X_h}{n} \ge c_h \cdot n^{-1/2^h}$ for all $n \ge 1$.
Consider, for any $n \ge 1$ and any $\ell > 0$, the event~$A_\ell$ that
 $W_h$ needs at least $\ell$ steps to terminate
 (not counting the steps of the spawned sub-runs)
 and that at least one of the spawned sub-runs needs at least $n$ steps to terminate.
Clearly, $\termt{X_h}(w) \ge n$ holds for all $w \in A_\ell$, so it suffices to
 find $c_h > 0$ so that for all $n \ge 1$ there is $\ell > 0$ with
 $\calP(A_\ell) \ge c_h \cdot n^{-1/2^h}$.
At least half of the steps of~$W_h$ are steps down,
 so whenever $W_h$ needs at least $2 \ell$ steps to terminate,
 it spawns at least $\ell$ sub-runs.
It follows:
\begin{align*}
 \calP(A_\ell)
 & \ge \calP(\text{$W_h$ needs at least $2 \ell$ steps}) \cdot
    \left( 1 - \left( \calP( \termt{X_{h-1}} < n ) \right)^\ell \right) \\
 & \ge \frac{c_1}{\sqrt{2 \ell}} \cdot
    \left( 1 - \left( 1 - c_{h-1} \cdot n^{-1/2^{h-1}} \right)^\ell \right) \qquad \text{(by induction hypothesis)}\\
\intertext{Now we fix $\ell := n^{1/2^{h-1}}$.
Then the second factor of the product above converges to $1 - e^{-c_{h-1}}$ for $n \to \infty$,
 so for large enough~$n$}
 \calP(A_\ell)
 & \ge \frac{c_1}{2} \cdot (1 - e^{-c_{h-1}}) \cdot n^{-1/2^h}\;.
\end{align*}
Hence, we can choose $c_h < \frac{c_1}{2} \cdot (1 - e^{-c_{h-1}})$ such that
 $\calP(A_\ell) \ge c_h \cdot n^{-1/2^h}$ holds for all $n \ge 1$.
\qed
\end{proof}

\smallskip
\noindent
\textbf{Acknowledgment.} The authors thank Javier Esparza for
useful suggestions.


\end{document}